\documentclass[sn-mathphys,Numbered]{sn-jnl}

\usepackage{lmodern}
\usepackage{graphicx}%
\usepackage{multirow}%
\usepackage{amsmath,amssymb,amsfonts}%
\usepackage{amsthm}%
\usepackage{mathrsfs}%
\usepackage[title]{appendix}%
\usepackage{xcolor}%
\usepackage{textcomp}%
\usepackage{manyfoot}%
\usepackage{booktabs}%
\usepackage{algorithm}%
\usepackage{algorithmicx}%
\usepackage{algpseudocode}%
\usepackage{listings}%
\usepackage{bm}
\usepackage{natbib}
\usepackage{longtable}
\usepackage{array}
\usepackage{hyperref}%
\usepackage{thmtools}
\usepackage{gensymb}
\usepackage[export]{adjustbox}
\usepackage{rotating}

\theoremstyle{plain}
\newtheorem{theorem}{Theorem}[section]

\theoremstyle{definition}

\theoremstyle{remark}

\newcommand{\quotes}[1]{`#1'}

\usepackage{lineno}

\begin{document}
\title[A semiparametric generalized exponential regression model with a principled distance-based prior]{A semiparametric generalized exponential regression model with a principled distance-based prior}

\author[1]{\fnm{Arijit} \sur{Dey}}\email{arijit.dey@duke.edu}

\author*[2]{\fnm{Arnab} \sur{Hazra}}\email{ahazra@iitk.ac.in}



\affil[1]{\orgdiv{Department of Statistical Science, Trinity College of Arts and Sciences, } \orgname{\\ Duke University}, \orgaddress{\city{Durham, NC}, \postcode{27708-0251}, \country{United States}}}

\affil*[2]{\orgdiv{Department of Mathematics and Statistics}, \orgname{Indian Institute of Technology Kanpur}, \orgaddress{\city{Kanpur}, \postcode{208016}, \country{India}}}

\abstract{The generalized exponential distribution is a well-known probability model in lifetime data analysis and several other research areas, including precipitation modeling. Despite having broad applications for independently and identically distributed observations, its uses as a generalized linear model for non-identically distributed data are limited. This paper introduces a semiparametric Bayesian generalized exponential (GE) regression model. Our proposed approach involves modeling the GE rate parameter within a generalized additive model framework. An important feature is the integration of a principled distance-based prior for the GE shape parameter; this allows the model to shrink to an exponential regression model that retains the advantages of the exponential family. We draw inferences using the Markov chain Monte Carlo algorithm and discuss some theoretical results pertaining to Bayesian asymptotics. Extensive simulations demonstrate that the proposed model outperforms simpler alternatives. The Western Ghats mountain range holds critical importance in regulating monsoon rainfall across Southern India, profoundly impacting regional agriculture. Here, we analyze daily wet-day rainfall data for the monsoon months between 1901--2022 for the Northern, Middle, and Southern Western Ghats regions. Applying the proposed model to analyze the rainfall data over 122 years provides insights into model parameters, short-term temporal patterns, and the impact of climate change. We observe a significant decreasing trend in wet-day rainfall for the Southern Western Ghats region.}




\maketitle

\section{Introduction}
\label{sec:Intro}

The Western Ghats region, a prominent mountain range along the western coast of India, plays a crucial role in shaping the climatic patterns and hydrological dynamics of Southern India. Known for its exceptional biodiversity, lush forests, and vital water resources, the Western Ghats has long captured the attention of researchers and environmentalists \citep{mathew2021spatiotemporal, venkatesh2021spatio}. Among the various climatic parameters that influence this ecologically significant region, rainfall is a crucial driver of its diverse ecosystems, water availability, and overall environmental health. The Western Ghats region, characterized by its rugged terrain and proximity to the Arabian Sea, experiences a unique and intricate rainfall pattern heavily influenced by monsoon dynamics \citep{varikoden2019contrasting}. Over the last century, this region has experienced notable climatic shifts due to global climate change and local human activities \citep{abe2013impacts, Veer2022impact}. Analyzing wet-day rainfall in this region during monsoon months over an extended period of 122 years using a flexible statistical model offers a valuable opportunity to gain insights into short and long-term trends, variability, and potential shifts in the monsoonal regime. While our methodology is motivated by these goals in the context of rainfall data analysis for the Western Ghats mountain range, it is also broadly applicable to analyzing rainfall and other types of datasets, including lifetime data, financial data, etc., for other parts of the globe.

Researchers have widely employed the exponential distribution to model rainfall data \citep{todorovic1974stochastic, hazra2018bayesian} in literature; its simplicity integrates seamlessly into hydrological and climatological frameworks. However, contemporary research increasingly recognizes the need for innovative probability distribution models to encompass complex real-world data patterns better. This realization has prompted the introduction of novel probability classes with far-reaching implications across diverse research domains. \cite{ahmad2019recent} provide an excellent overview of newly developed distributions. A notable collection of models is generalized distributions, which have gained attention from both practical and theoretical statisticians for their adaptability to various datasets. Some examples include the Marshall-Olkin generalized exponential distribution \citep{marshall1997new}, generalized inverse Gaussian distribution \citep{jorgensen1982statistical}, generalized Rayleigh distribution \citep{kundu2005generalized}, etc. \citep{tahir2016compounding} examine these distributions comprehensively.

\cite{gupta1999theory} introduced another crucial generalized distribution called the generalized exponential (GE) distribution, which emerges as a specific case within the three-parameter exponentiated-Weibull model. The GE distribution has two parameters- a shape and a rate (or scale, defined as the inverse of rate) parameter. This distribution boils down to an exponential distribution when the shape parameter is one. Thus, an additional shape parameter expands the capabilities of the exponential distribution, making it more adaptable to various datasets. Since its introduction, many researchers have integrated substantial advancements in exploring different properties, estimation strategies, extensions, and applications of this distribution. For instance, \cite{gupta2001exponentiated} found the efficacy of GE distribution compared to gamma or Weibull distributions, whereas \cite{gupta2001generalized} discussed different methods of estimating the parameters of GE distribution. \cite{jaheen2004empirical}, \cite{raqab2005bayesian}, and \cite{kundu2008generalized} explored Bayesian estimation and prediction methods in this context. \cite{gupta2007generalized} reviewed the existing results and discussed some new estimation methods and their characteristics. Numerous researchers have modeled the experimental data using GE distribution across several disciplines, like meteorological studies \citep{madi2007bayesian, hazra2024minimum}; flood frequency analysis \citep{markiewicz2015generalized}; reliability analysis \citep{aslam2007economic}; lifetime analysis \citep{cota2009alternative}; risk analysis \citep{sarhan2007analysis}. However, as of our knowledge, regression-type models based on the GE distribution have never been proposed in the literature. Besides, searching using the phrase \textit{Generalized exponential regression} in Google Scholar or ChatGPT does not lead to any relevant papers. Existing works on the GE distribution focus mainly on its properties or distributional parameter estimation, but do not extend to regression models where the parameters are modeled through covariates. In particular, there is no evidence of either frequentist or Bayesian estimation methods for a parametric or semiparametric GE regression framework. 

Rainfall data collected over a century are inherently nonstationary. Here, modeling the temporal trend using traditional parametric regression would struggle to capture the intricate and evolving short-term temporal patterns. In this context, a semiparametric regression setup emerges as a promising approach. In the existing literature, many researchers have delved into applying semiparametric regression techniques for analyzing rainfall patterns \citep{wigena2015semiparametric, nguyen2020probabilistic}. While a generalized linear model (GLM) assumes the link function to be a linear combination of the covariates, the more flexible generalized additive models \citep[GAM,][]{hastie1990generalized} allow the link function to be a sum of nonlinear smooth functional forms of the underlying covariates. We generally model each smooth function in GAMs as a linear combination of basis functions like cubic B-splines. Instead of estimating the entirely unknown function, following a finite truncation (hence semiparametric) of the number of B-splines, we draw inferences based on basis function coefficients \citep{ghosal2017fundamentals}. Henceforth, instead of GAM, we use the term \textit{semiparametric regression}, which is common in Bayesian nonparametrics. The rate parameter of the GE distribution is always positive, and hence, it would be reasonable to model the log-rate in a semiparametric regression framework.

Within the Bayesian methodology, priors hold a pivotal role in inference, and the literature provides a diverse spectrum of prior distributions utilized for regression coefficients in semiparametric regression frameworks. For instance, a Gaussian prior was employed by \cite{gelfand2003bayesian}, while \cite{lee2018bayesian} opted for a Laplace prior. \cite{li2014bayesian} utilized Zellner's $g$-prior, while \cite{fahrmeir2001bayesian} considered flat priors, and \cite{koop2004bayesian} used the Normal-Gamma prior. On the other hand, the gamma distribution has consistently been considered the most natural prior choice for the shape parameter of the GE distribution; authors who introduced the GE distribution chose a gamma prior for the shape parameter in \cite{kundu2008generalized} as well. Besides \cite{raqab2005bayesian} and \cite{kim2010bayesian} also employed a gamma prior for the shape parameter. However, the literature demonstrates that a handful of alternative prior choices have also been utilized. For example, \cite{naqash2016bayesian} employed a Jeffrey's prior, indicating their preference for an objective prior, and \cite{dey2010bayesian} opted for a non-informative prior in their study.

The Penalized Complexity (PC) prior, introduced by \cite{simpson2017penalising}, has emerged in recent literature, which mitigates the model complexity through penalization. In cases where a model extends from a simpler foundational model by incorporating an additional parameter, this type of prior becomes applicable; it penalizes the escalation in model complexity that arises when favoring the extended model over its more straightforward counterpart. Existing literature encompasses instances of this approach across various models \citep{van2021principled}. \cite{ventrucci2016penalized} developed PC priors for estimating the effective degrees of freedom in Bayesian penalized splines (P-splines), while \cite{ordonez2023penalized} discussed a PC prior for the skewness parameter of the power links family, and \cite{sorbye2017penalised} proposed interpretable and comprehensive PC priors for the coefficients of a stationary autoregressive process.

In this paper, along with proposing a semiparametric Bayesian GE regression model where we build the rate parameter in a generalized additive model framework (in a log-scale), we employ the PC prior for the GE shape parameter, which allows the GE regression to shrink towards an exponential regression. Thus, the exponential distribution is considered the base model for the GE distribution. In several practical examples \citep{hazra2018bayesian, hazra2024robust}, the exponential distribution is found to be a reasonable model and enjoys several benefits of being a member of the exponential family; thus, shrinking the GE distribution to its base model through shrinking the shape parameter to one is justified. On the other hand, we opt for the independent Gaussian priors for the regression coefficients. We draw inferences using the Markov chain Monte Carlo (MCMC) algorithm; here, conjugate priors are not available for the model parameters, and thus, we update them using Metropolis-Hastings steps. We further discuss some theoretical results related to Bayesian asymptotics. We conduct a thorough simulation study by simulating 1000 datasets from each combination of the model generating and model fitting scenarios, and we compare the performances of parametric and semiparametric Bayesian GE regression models under the conventional gamma prior choices for the GE shape parameter, along with our proposed one. We study the coverage probabilities for the shape parameter and the rate functions and compare these two models using Watanabe–Akaike information criterion (WAIC)  \citep{watanabe2010asymptotic}. We implement the proposed methodology to the daily wet-day precipitation spanning from 1901 to 2022 in different regions of the Western Ghats mountain range, using the year as a covariate and wet-day precipitation as a response. We study the convergence and mixing of the MCMC chains and compare different model fits in terms of WAIC.

The paper is structured as follows. Section \ref{sec:GE} delves into the necessary background about the GE distribution, thoroughly examining its properties. In Section \ref{sec:Regression}, we introduce the GE regression model. Proceeding to Section \ref{sec:Prior specification}, we concentrate on delineating the prior specifications for the regression model, including introducing a principled distance-based prior for the shape parameter of the GE distribution. Bayesian inference under small and large sample scenarios is addressed in Section \ref{sec:Inference}. Section \ref{sec:Simulation} presents the outcomes of the simulation study, while Section \ref{sec:Data_Application} discusses an exploratory data analysis that justifies our semiparametric GE model assumption for the wet-day precipitation data, followed by the results obtained from our proposed model and some simpler alternatives. Finally, Section \ref{sec:Conclusion} summarizes our findings and contributions.


\section{Background: Generalized Exponential (GE) Distribution}
\label{sec:GE}

We say a random variable $Y$ follows a GE distribution if its cumulative distribution function (CDF) is given by
$$F(y;\alpha,\lambda) = \left(1 - e^{-\lambda y}\right)^\alpha, \hspace{10pt} y , \alpha, \lambda>0,$$
where $\alpha$ is the shape parameter and $\lambda$ is the rate parameter. The corresponding probability density function (PDF) is given by 
\begin{equation}
\label{eq:pdf_GE}
    f(y;\alpha, \lambda) = \alpha \lambda \left(1 - e^{-\lambda y}\right)^{\alpha - 1} e^{-\lambda y}, \hspace{10pt} y , \alpha, \lambda>0.
\end{equation}

The GE distribution is a more complex model than the exponential distribution, as it incorporates an extra shape parameter. Both models coincide when $\alpha$ = 1.


\subsection{Properties of GE}

The hazard function of the GE distribution is given by $$h(y;\alpha, \lambda) = \frac{f(y;\alpha, \lambda)}{1 - F(y;\alpha, \lambda)} = \frac{\alpha \lambda \left(1 - e^{-\lambda y} \right)^{\alpha-1} e^{-\lambda y}}{1 - \left(1 - e^{-\lambda y} \right)^\alpha}, \hspace{10pt} y > 0.$$ The GE distribution has an increasing or decreasing hazard rate depending on the value of the shape parameter. The hazard function decreases for $\alpha < 1$, remains constant for $\alpha = 1$, and increases for $\alpha > 1$. The moment generating function (MGF) of the GE distribution is given by $$M_Y(t) = \frac{\Gamma (\alpha+1) \Gamma(1- t/\lambda )}{\Gamma(1+\alpha-t/\lambda)}, \hspace{5pt} 0 \le t < \lambda,$$ and differentiating the log of the MGF with respect to $t$ repeatedly and then setting $t = 0$, we get the expectation, variance, and skewness of GE distribution as
\begin{eqnarray}
\nonumber   \textrm{E}(Y) &=& \lambda^{-1} \left[ \psi(\alpha + 1) - \psi(1)\right], \\
\nonumber   \textrm{V}(Y) &=& \lambda^{-2}\left[ \psi^{(1)}(1) - \psi^{(1)}(\alpha + 1) \right], \\
\nonumber  \text{Skewness}(Y) &=& \left[ \psi^{(2)}(\alpha + 1) - \psi^{(2)}(1) \right] \bigg/ \left[\psi^{(1)}(1) - \psi^{(1)}(\alpha + 1) \right]^{3/2}, \nonumber
\end{eqnarray}
where $\psi^{(m)}(z) = \frac{\partial^m}{\partial z^m}\psi(z) = \frac{\partial ^{m+1}}{\partial z^{m+1}}\ln \Gamma(z)$ is the polygamma function of order $m$; for $m = 0$, it  denotes the digamma function. 

\begin{figure}[t]
    \centering
    \includegraphics[width=0.48\textwidth]{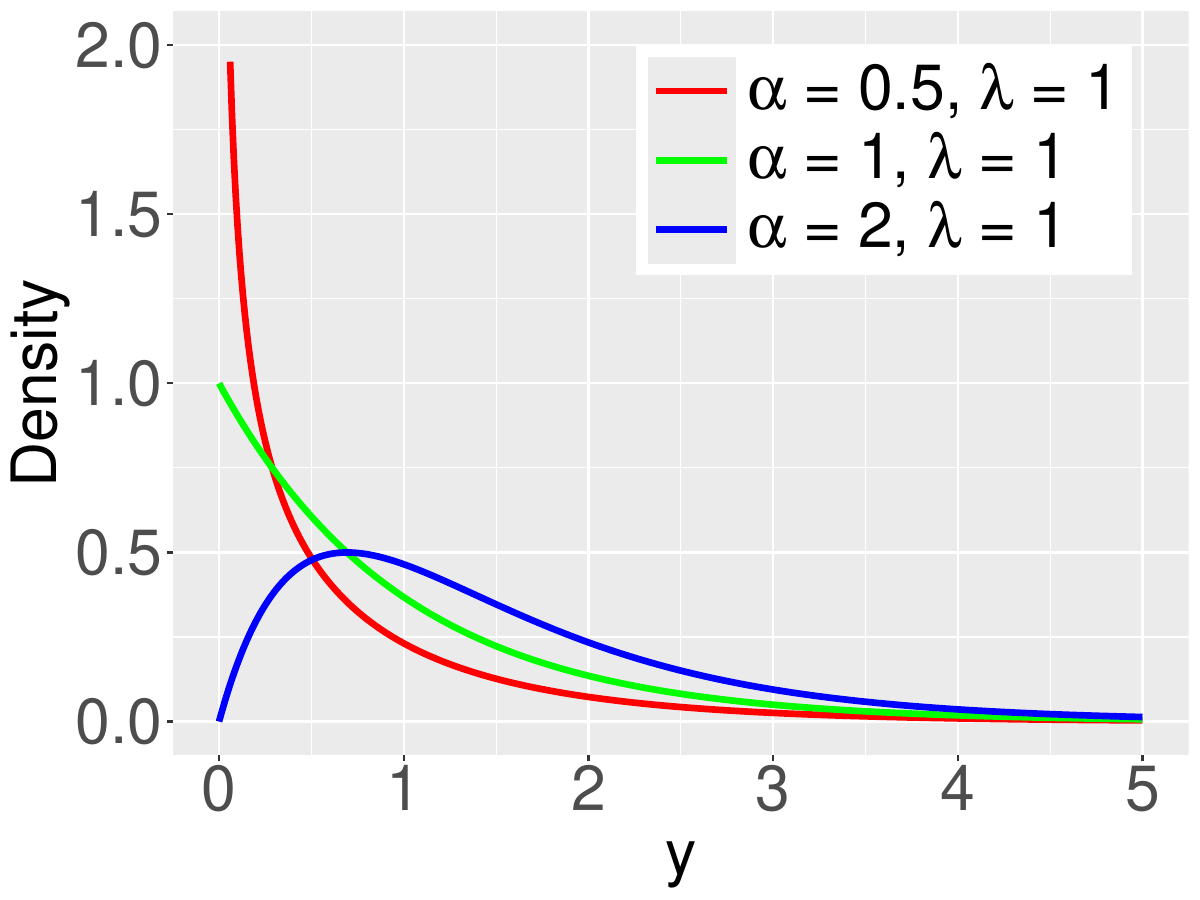}
    \includegraphics[width=0.48\textwidth]{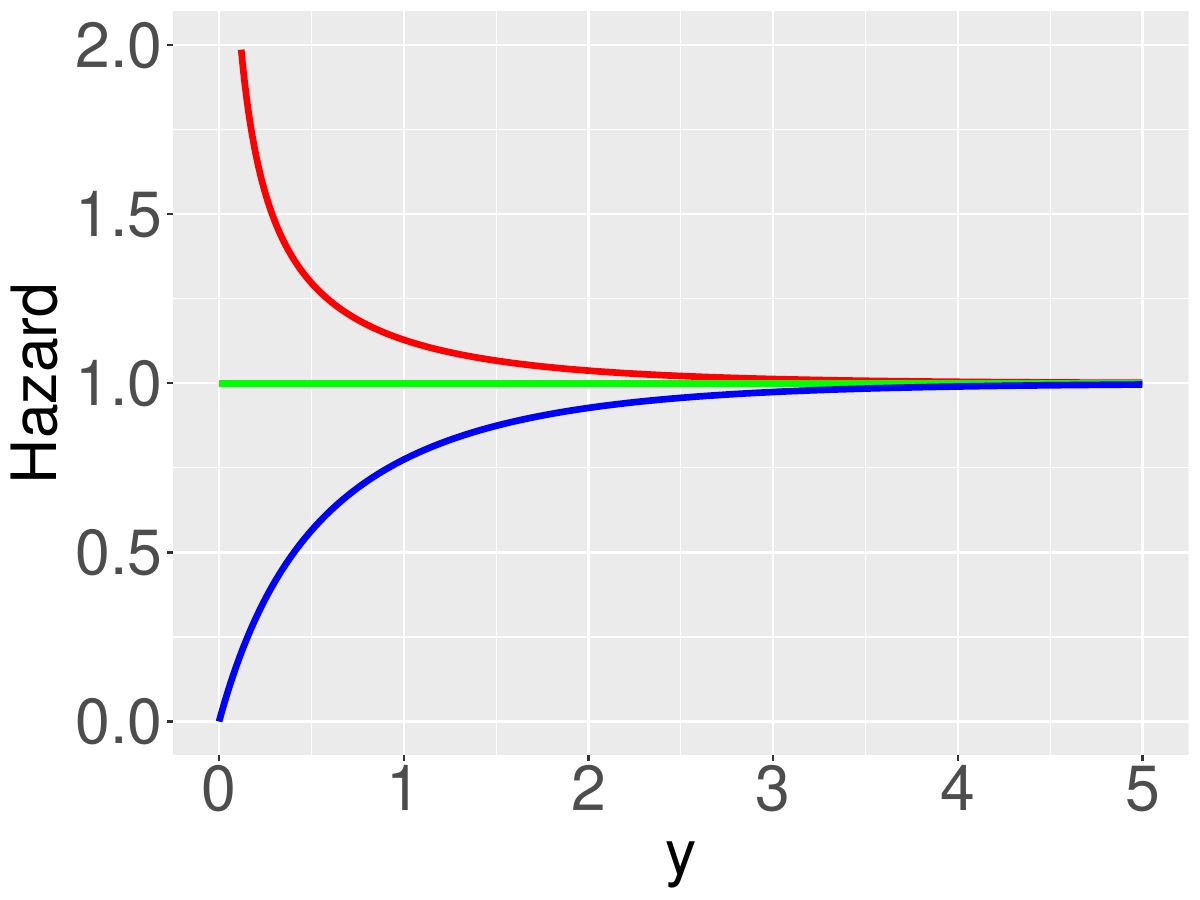}
    \includegraphics[width=0.32\textwidth]{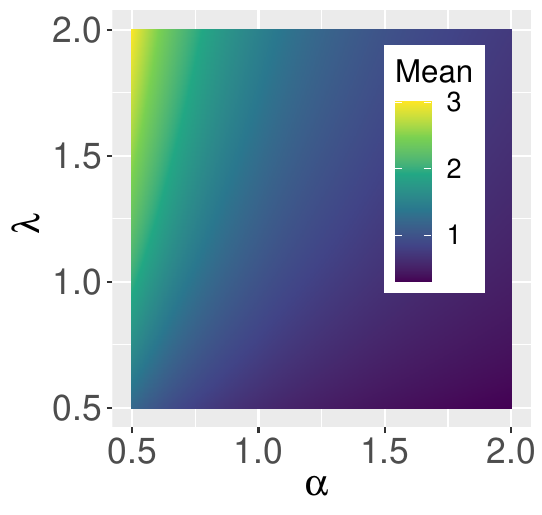}
    \includegraphics[width=0.32\textwidth]{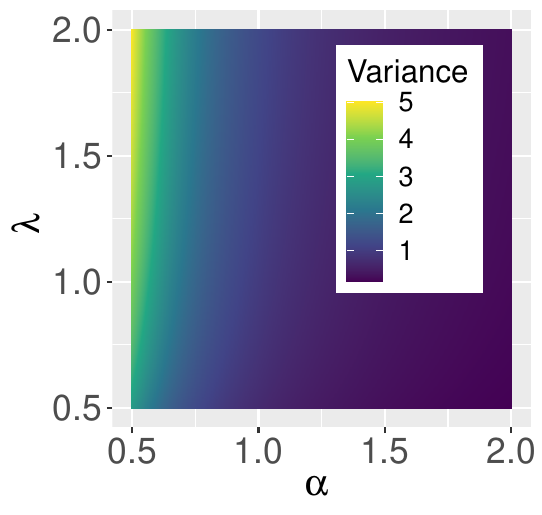}
    \includegraphics[width=0.32\textwidth]{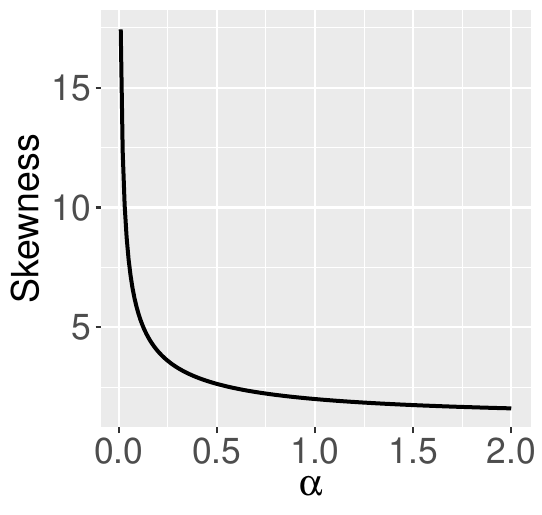}
    \caption{Generalized Exponential probability density function (top-left), hazard function (top-right), mean (bottom-left), variance (bottom-middle), and skewness (bottom-right) functions. Both top panels share the same legend.}
    \label{fig:GE_graphs}
\end{figure}

Figure \ref{fig:GE_graphs} sheds light on different aspects of the GE distribution, e.g., PDF, hazard function, mean, variance, and skewness. The top-left panel of Figure \ref{fig:GE_graphs} shows that for $\alpha < 1$, the curve depicting the PDF of the GE distribution has an asymptote at the Y-axis and then decreases exponentially and monotonically as we move across the positive real line. With $\alpha = 1$, GE coincides with the exponential distribution, thus having the mode at zero (with value $= \lambda$) and gradually decreasing similarly as the previous case. When $\alpha > 1$, the curve initiates at zero, then increases over a range of values, and eventually decreases monotonically, having a unique mode at $\log(\alpha)/\lambda$. As mentioned earlier, the top-right panel of Figure \ref{fig:GE_graphs} shows that the hazard function is monotonically decreasing when $\alpha < 1$, monotonically increasing when $\alpha > 1$, and constant (the value being $\lambda =1$) when $\alpha = 1$. The mean and variance of GE behave somewhat similarly. From the bottom-left and the bottom-middle panel of Figure \ref{fig:GE_graphs}, we see that for a fixed value of $\alpha$, both mean and variance decrease with increasing $\lambda$, and for a fixed value of $\lambda$, both increase as $\alpha$ increases. On the other hand, the skewness of the GE distribution depends only on the shape parameter and decreases exponentially with increasing $\alpha$ (bottom-right panel of Figure \ref{fig:GE_graphs}).

\section{Generalized Exponential (GE) Regression}
\label{sec:Regression}

\subsection{Parametric GE Regression}
\label{subsec:theo_regression_parametric}
The observations often do not satisfy the assumption of being independently and identically distributed. For example, the rainfall data observed across 122 years in the Western Ghats mountain range are unlikely to be identically distributed due to several potential short-term and long-term factors, such as El Niño and global warming. However, different rainfall events across the years can be safely assumed to be independent, which is a common assumption in the environmental statistics literature. We thus can consider a regression model where the response variable $Y$ follows a GE distribution, and the relationship between $Y$ and the covariates $\bm{X}= (X_1, \ldots, X_P)'$ is represented through a linear predictor $\eta$, i.e., a linear combination of the covariates with associated regression coefficients, given by 
\begin{equation}
    \label{eq:linear_predictor}
    \eta(\bm{X}) = \phi_0 + \phi_1 X_1 + \phi_2 X_2 + \dots + \phi_P X_P,
\end{equation}
where $\bm{\phi} = (\phi_0, \phi_1, \ldots, \phi_p)'$ is the vector of regression coefficients.

Here, the shape parameter $\alpha$ is considered an inherent property of the distribution that characterizes the shape and asymmetry of the distribution, allowing for a more flexible modeling approach compared to a standard linear regression with a Gaussian error component. On the other hand, the rate parameter $\lambda$ is the parameter of interest in the regression model, which captures the association between the covariates and the response variable. Moreover, given that the rate parameter of the GE distribution is positive, we relate it to the linear predictor from \eqref{eq:linear_predictor} using a link function designed to ensure the rate parameter always stays positive. Thus, for $n$ observed responses and the corresponding covariate vectors $\{Y_i, \bm{X}_i; i=1, \ldots, n\}$, we conceptualize the GE regression model as 
$$Y_i | \bm{X}_i = \bm{x}_i \overset{\textrm{Indep}}{\sim} \textrm{GE}(\alpha, \lambda_i),~~ i=1, \ldots, n,$$
where $g(\lambda_i) =\eta(\bm{x}_i)$, with $g(\cdot)$ representing an appropriate link function and $\bm{x}_i = (x_{i1}, x_{i2}, \dots, x_{iP})'$. A natural choice for $g(\cdot)$ is $g(\lambda) = \log(\lambda)$, which ensures $\lambda_i > 0$ for all $i$.


\subsection{Semiparametric GE Regression}
\label{subsec:theo_regression}

The parametric GE regression model introduced in Section \ref{subsec:theo_regression_parametric} aims to capture the linear relationship between the response (suitably transformed) and the covariates. However, such an assumption may be practically unreliable; for example, assuming the mean or median rainfall to vary linearly across 122 years for the Western Ghats range is unlikely, and several short-term decadal patterns can be present in the data. In contrast, a nonparametric regression model assumes no specific form for the relationship between the response and covariates, allowing flexibility based on data-derived information. While this approach allows for more flexible modeling, it is computationally intensive, less interpretable, and can be affected by the \textit{curse of dimensionality}. Semiparametric regression integrates the above two approaches, allowing us to have the best of both regimes; it incorporates the interpretability of the parametric setup and the flexibility of the nonparametric setup. 

In linear models or generalized linear models (GLM) that fall under the parametric setup, we assume the conditional mean of the distribution of the response variable is linked with the linear predictor through a linear combination of the covariates or their functions. The generalized additive model \citep{hastie1990generalized} setup extends this domain of regression models by introducing the nonparametric component in the linear predictors. In this setup, the most general formulation of the linear predictor can be given as 
\begin{equation}
    \label{eq:Semi-parametric_model}
    \eta(\bm{x}) = \sum_{p=1}^P f_p(x_p),
\end{equation}
where $f_j$'s are smoothing functions of continuous covariates and $\bm{x} = (x_1, x_2, \dots, x_P)'$. Under a purely nonparametric scenario, each of $f_j$'s allows an infinite basis function expansion, while most semiparametric methods assume they can be expressed as a linear combination of finite basis functions, often denoted as
\begin{equation}
    \label{eq:Basis splines}
    f_p(z) = \sum_{k = 1}^{K_p} \beta_{p,k} B_{p,k}(z), \hspace{8pt}p = 1, 2, \dots, P,
\end{equation}
where $B_{p,k}(\cdot)$'s are known basis functions and $\beta_{p,k}$ are unknown basis function coefficients that determine the shape of the smoothing function $f_p(z)$. A basis expansion of $M$-many terms can match the true curve $f_p(\cdot)$ at any $M$ points $X_1, \dots, X_M$ in the range of covariates. Hence, increasing $M$ gives us an arbitrarily flexible model, and cross-validation or information criterion-based choice of $M$ is necessary for practical purposes.

In this study, we employ a semiparametric model akin to \eqref{eq:Semi-parametric_model} for the rate parameter of the GE distribution. With a covariate vector comprising $P$ components and the appropriate logarithmic link function, the regression model takes the form
\begin{equation}
\label{eq:GE_semi_par_model}
    Y_i | \bm{X}_i = \bm{x}_i \overset{\textrm{Indep}}{\sim} \textrm{GE}\big(\alpha, \lambda(\bm{x}_i)\big) \text{ with } \log \left\lbrace\lambda(\bm{x}_i)\right\rbrace = \sum_{p = 1}^{P} \sum_{k=1}^{K_p} \beta_{p,k} B_{p,k} (x_{ip}),
\end{equation}
where $B_{p,k}(\cdot)$'s are cubic B-splines and $\beta_{p,k}$'s are the spline coefficients representing the weights assigned to the corresponding spline functions. Here, a cubic B-spline is a piecewise-defined polynomial cubic function that is defined on a set of knots or control points, taking the form $B_s(x) = (x - v_s)_{+}^3$ where $v_s$'s are the fixed knots that span the range of $x$ and \quotes{+} denotes the positive part. In our model, we choose equidistant knots. Denoting $\bm{\beta}_p = (\beta_{p,1}, \ldots, \beta_{p,K_p})', p=1,\ldots,P$ and combining them, we have $\bm{\beta} = (\bm{\beta}'_1, \ldots, \bm{\beta}'_P)'$. Further, for the B-spline components, we denote $\bm{B}_p(\bm{x}_i) = [B_{p,1}(x_{i,p}), \ldots, B_{p,K_p}(x_{i,p})]',p=1,\ldots,P$ and $\bm{B}(\bm{x}_i) = [\bm{B}_1(\bm{x}_i)', \ldots, \bm{B}_P(\bm{x}_i)']'$. By an abuse of notation, for $K=\sum_{p = 1}^{P} K_p$, we henceforth denote $\bm{B}(\bm{x}_i) = (b_{i,1}, \ldots, b_{i,K})'$ and $\bm{\beta}=(\beta_1,\ldots, \beta_K)'$ for $i=1,\ldots,n$.

\section{Prior Specification}
\label{sec:Prior specification}

We first discuss the prior specification for the rate parameter of the GE distribution. If the observations are independently and identically distributed with the same rate parameter $\lambda$, we can formulate an explicit prior for $\lambda$. Alternatively, in the case of a parametric GE regression model, we can choose independent priors for the elements of $\bm{\phi}$ in \eqref{eq:linear_predictor}. In the case of a generalized additive model setup, as in \eqref{eq:Semi-parametric_model}, without any finite basis function expansion, we need to choose priors for $f_p$'s; a standard prior choice in the literature is a Gaussian process \citep{ghosal2017fundamentals}. In semiparametric Bayesian regression, independent prior distributions are explicitly specified for the spline coefficients $\bm{\beta} = (\beta_1, \beta_2, \dots, \beta_K)'$. This paper considers independent weakly-informative Gaussian priors for $\beta_k$'s; specifically $\beta_k \overset{\textrm{IID}}{\sim} \textrm{N}(0,10^2)$.

We employ a newly developed class of priors for the shape parameter of the GE distribution. In cases where a model is constructed based on a simpler base model, the chosen prior should accurately reflect the characteristics of the model considered and capture its departure from the base model. This type of prior construction is founded upon the work of \cite{simpson2017penalising}, who introduced the Penalized Complexity (PC) prior. The PC prior is a model-based prior that imposes a penalty on the deviation of the model of consideration from its simpler base version at a logarithmic constant rate. The following subsection discusses the PC prior for $\alpha$.


\subsection{Penalized Complexity (PC) prior}
\label{subsec:PC prior}

The PC prior is an informative proper prior that exhibits robustness properties of high quality and invariance under reparameterization. It aims to penalize the complexity that arises when we move from a simpler base model to a more complex one, thereby preventing overfitting and adhering to \textit{Occam’s razor principle} \citep{simpson2017penalising}. By using the PC prior, we uphold the \textit{principle of parsimony}, which suggests a preference for simpler models until sufficient evidence supports more complex alternatives.

The PC prior is established based on the statistical difference between the proposed complex model and its base model. We quantify this distance using the Kullback-Leibler divergence (KLD) \citep{kullback1951information}, which essentially measures the information loss when we substitute a complex model with PDF $f$ with its simpler version with PDF $g$. The exponential distribution is a natural choice as the appropriate base model for the GE distribution. Hence, for our purposes, we consider $f$ and $g$ as the GE and exponential densities, respectively.

For two continuous distributions with probability distribution functions $f$ and $g$ defined over the same support, KLD is defined as
\begin{equation}
    \label{eq:KLD}
    \text{KLD}(f \parallel g) = \int_{-\infty}^{\infty} f(y) \log \left( \frac{f(y)}{g(y)} \right) \, dy \,,
\end{equation} 
where we define the distance between the two models by the \quotes{unidirectional} distance function $d(f \parallel g) = \sqrt{2\text{KLD}(f \parallel g)}$ \citep{simpson2017penalising}. The absence of symmetry in KLD is not a concern within this context. Our main focus is on quantifying the additional complexity that arises from employing the intricate model rather than the other way around.

The main idea of the PC prior involves assigning priors to the distance between two models rather than directly on the model parameters, and then by employing a change-of-variables approach, one can obtain a prior distribution for the parameter of interest. In our context, while constructing the PC prior for the shape parameter $\alpha$, we take this distance as a function of $\alpha$, i.e., $d(\alpha) = \sqrt{2\text{KLD}(\alpha)} \equiv \sqrt{2\text{KLD}(f \parallel g)}$.

To incorporate the fact that the prior should have a decaying nature as a function of the distance between the two models, we take the constant rate penalization assumption and construct the PC prior by assigning an exponential prior to the distance, i.e.,  $d(\alpha) \sim \text{Exp}(\alpha_0)$ with $\alpha_0 > 0$; this gives us the PC prior for $\alpha$ as
\begin{equation}
    \label{eq:PC}
    \pi(\alpha) = \alpha_0 \exp[-\alpha_0 d(\alpha)] \left\vert \frac{\partial d(\alpha)}{\partial \alpha} \right\vert,
\end{equation}
where $\alpha_0$ is a user-defined quantity that controls the prior mass at the tail; this is a user-defined quantity that characterizes how informative we want the PC prior to be. We achieve this by imposing a condition $\Pr[d(\alpha) > U ] = \xi$, where $U$ is the upper bound of the tail-event and $\xi$ is the weight of the event \citep{simpson2017penalising}. 

\subsection{PC prior for the shape parameter of the GE distribution}
\label{subsec:PC prior GE}

The theorems in this section introduce the KLD between our complex model (the GE density $f$) and its base model (the exponential density $g$) and the PC prior of the shape parameter $\alpha$.

\vspace{0.3cm}

\begin{theorem}
\label{th:KLD_GE}
The KLD between the GE density function in \eqref{eq:pdf_GE} and the density function of the exponential distribution with rate $\lambda$ is given by $$\textrm{KLD}(\alpha) = \log (\alpha) + 1 / \alpha -1.$$
\end{theorem}

\begin{theorem}
    \label{th:PC_prior}
    The PC prior, with hyperparameter $\alpha_0$, for the shape parameter $\alpha$ is  $$\pi (\alpha) = \frac{\alpha_0}{2}  \exp\left(-\alpha_0 \,  \sqrt{2\log(\alpha) + \frac{2(1-\alpha)}{\alpha}}\right)  \left(2\log(\alpha) + \frac{2(1-\alpha)}{\alpha} \right)^{-1/2}  \left| \frac{1}{\alpha} - \frac{1}{\alpha^2}\right|, ~~ \alpha, \alpha_0 > 0.$$
\end{theorem}

Proof for Theorem \ref{th:KLD_GE} is given in Appendix \ref{apndx:1} and Theorem \ref{th:PC_prior} follows directly from \eqref{eq:PC} and the expression $d(\alpha) = \sqrt{2\text{KLD}(\alpha)}$, where $\text{KLD}(\alpha)$ is specified in Theorem \ref{th:KLD_GE}. Moreover, the density includes a scaling factor that ensures $\int_{0}^{\infty}\pi(\alpha) d \alpha = 1$.

\begin{figure}[ht]
    \centering
    \includegraphics[width=0.6\textwidth]{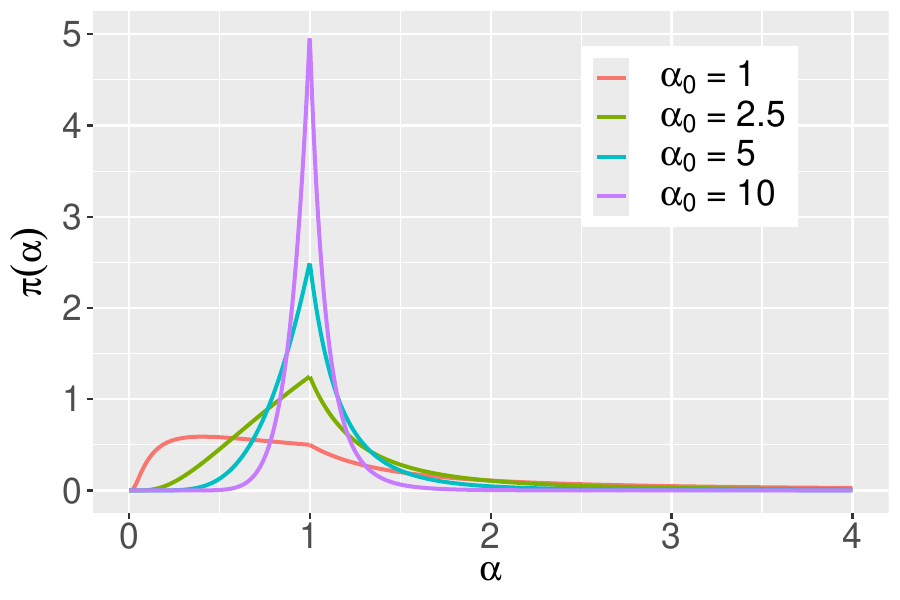}
    \caption{The PC prior for the shape parameter of the GE distribution for different choices of the hyperparameter $\alpha_0$.}
    \label{fig:PC_pdf}
\end{figure}

Figure \ref{fig:PC_pdf} illustrates $\pi(\alpha)$ for different hyperparameter specifications $\alpha_0$. We notice a proportional relationship between the value of $\alpha_0$ and the extent of contraction to the base model. As the value of $\alpha_0$ decreases, the tails become heavier, resulting in reduced contraction towards the base model. We also observe that for $\alpha_0 \leq 4/3$, the mode of the density function occurs at a value of $\alpha$ less than one, but for $\alpha_0 \geq 4/3$, the mode is at $\alpha = 1$. While one might expect the prior would have mode at $\alpha = 1$ irrespective of the value of $\alpha_0$, we do not necessarily need the mode at $\alpha$ = 1, and rather, we should rely on a prior that is consistent with the principles of PC prior.


\section{Bayesian Inference}
\label{sec:Inference}

This paper employs a Bayesian estimation method to infer and quantify the uncertainty surrounding the parameters of interest. In this context, the likelihood function based on $n$ observations from the GE distribution under the regression setting from \eqref{eq:GE_semi_par_model} is given by
\begin{equation}
    \label{eq:GE_lklhd_reg}
    L(\alpha, \bm{\beta} | \bm{y})  = \prod_{i = 1}^{n} f\big(y_i; \alpha, \lambda(\bm{x}_i)\big), 
\end{equation}
where $\bm{y} = (y_1, y_2, \dots, y_n)'$ is the observed data vector, $f(y; \alpha, \lambda)$ is the GE density function in \eqref{eq:pdf_GE} and $\lambda(\bm{x_i})$ taking form as given in \eqref{eq:GE_semi_par_model}. Also, let $\pi(\alpha)$ and $\pi(\bm{\beta})$ denote the specified mutually independent priors for the parameters $\alpha$ and $\bm{\beta}$.

Combining the priors for $\alpha$ and $\bm{\beta}$, and the likelihood function as given in \eqref{eq:GE_lklhd_reg}, we obtain the joint posterior distribution as $\pi(\alpha, \bm{\beta} | \bm{y}) \propto L(\alpha, \bm{\beta} | \bm{y}) \times \pi(\alpha) \times \pi(\bm{\beta})$ from which Bayesian inference is facilitated. However, the explicit form of the marginal posterior density of the parameters is not analytically tractable, leading to employing simulation-based techniques such as MCMC methods or numerical approximation methods like Integrated Nested Laplace Approximations (INLA), introduced by \cite{rue2009approximate}. This paper employs MCMC techniques for parameter inference, specifically utilizing the adaptive Metropolis-Hastings algorithm within Gibbs sampling. We iteratively adjust the variance of the proposal distribution within the burn-in phase so that the acceptance rate remains between 0.3 and 0.5. We initiate the MCMC chains with an initial value of 1 for $\alpha$ and the maximum likelihood estimate for $\bm{\beta}$ as calculated under $\alpha$ = 1. In our implementation, we update the model parameters one at a time within Gibbs sampling.

Furthermore, describing the asymptotic (as $n \uparrow \infty$ and keeping $K$ fixed) distribution of the posterior estimates for the parameters $\alpha$ and $\bm{\beta}$ is feasible using the Bernstein-von Mises theorem. To gauge the level of uncertainty linked with these parameter estimations, we investigate the asymptotic variance of the parameters, which is encapsulated by the inverse of the information matrix. 

\subsection{Bayesian Asymptotics}

In this subsection, we first discuss a result related to posterior consistency and further characterize the asymptotic posterior distribution for the proposed model. The following results do not assume the entire function $\lambda(\cdot)$ in \eqref{eq:GE_semi_par_model} to be unknown; rather, assuming the linear representation of $\log\{ \lambda(\cdot) \}$ in \eqref{eq:GE_semi_par_model} to be true, only focuses on the large sample results for $\alpha$ and $\bm{\beta}$. Suppose we denote the collection of the first $n$ observations by $\mathcal{Y}_n = \{Y_1, \ldots, Y_n\}$, the full vector of model parameters by $\bm{\theta} = (\alpha, \bm{\beta}')'$, the corresponding vector of true parameter values by $\bm{\theta}_{\ast} = (\alpha_*, \bm{\beta_{*}'})'$, and the maximum likelihood estimator of $\bm{\theta}$ based on $\mathcal{Y}_n$ by $\hat{\bm{\theta}}_n = \underset{\bm{\theta}}{\arg\max} \, \mathbb{P}(\mathcal{Y}_n|\bm{\theta})$. Here, the Fisher information matrix is defined as 
$$\mathcal{I}_n (\bm{s}) = \mathbb{E} \left[ \left. \left( \frac{\partial}{\partial\bm{\theta}} \log \mathbb{P} \left( \mathcal{Y}_n | \bm{\theta} \right) \right)^2 \right|_{\bm{\theta} = \bm{s}} \, \right] = - \mathbb{E} \left[ \left. \left( \frac{\partial^2}{\partial\bm{\theta}^2} \log \mathbb{P} \left( \mathcal{Y}_n | \bm{\theta} \right) \right) \right|_{\bm{\theta} = \bm{s}} \, \right],$$ 
where the expectation is calculated with respect to the data.

\begin{theorem}
\label{th:post_cons}
    The posterior distribution of $\bm{\theta} = (\alpha, \bm{\beta}')'$ is consistent, i.e., the posterior concentrates in the neighborhood of the true parameter value $\bm{\theta}_{\ast}$ as $n \uparrow \infty$.
\end{theorem}

The proof of Theorem \ref{th:post_cons} follows from Doob's theorem \citep{doob1949application} and it is provided in Appendix \ref{apndx:2}. 


\begin{theorem}
\label{th:Bayes_asymp}
Under sufficient regularity conditions \citep[Section 2.2,][]{patriota2012q}, we have
\begin{equation}
    \label{eq:Bayesian_aysmptotic}
    \sqrt{n} ( \bm{\theta} - \hat{\bm{\theta}}_n ) | \mathcal{Y}_n \xrightarrow[]{d} N_{K+1} \left(0, \widetilde{\mathcal{I}} (\bm{\theta}_\ast)^{-1} \right),
\end{equation}
where $\widetilde{\mathcal{I}} (\bm{\theta}) = \lim_{n \uparrow \infty} n^{-1} \mathcal{I}_n (\bm{\theta})$. Alternative notations for $\widetilde{\mathcal{I}} (\bm{\theta})$ and $\mathcal{I}_n (\bm{\theta})$ are $\widetilde{\mathcal{I}}(\alpha, \bm{\beta})$ and $\mathcal{I}_n(\alpha, \bm{\beta})$, respectively. The Fisher information matrix under our GE regression setup is  
\begin{equation}
    \label{eq:GE_Information_matrix_reg_exp}
    \mathcal{I}_n(\alpha, \bm{\beta}) = \begin{pmatrix}
I_{\alpha,\alpha}^{(n)} & I_{\alpha,1}^{(n)} & \cdots & I_{\alpha,K}^{(n)} \\
I_{1,\alpha}^{(n)} & I_{1,1}^{(n)} & \cdots & I_{1,K}^{(n)} \\
\vdots  & \vdots  & \ddots & \vdots  \\
I_{K,\alpha}^{(n)} & I_{K,1}^{(n)} & \cdots & I_{K,K}^{(n)} 
\end{pmatrix} = -
\begin{pmatrix}
\mathbb{E}(J_{\alpha,\alpha}^{(n)}) & \mathbb{E}(J_{\alpha,1}^{(n)}) & \cdots & \mathbb{E}(J_{\alpha,K}^{(n)}) \\
\mathbb{E}(J_{1,\alpha}^{(n)}) & \mathbb{E}(J_{1,1}^{(n)}) & \cdots & \mathbb{E}(J_{1,K}^{(n)}) \\
\vdots  & \vdots  & \ddots & \vdots  \\
\mathbb{E}(J_{K,\alpha}^{(n)}) & \mathbb{E}(J_{K,1}^{(n)}) & \cdots & \mathbb{E}(J_{K,K}^{(n)}) 
\end{pmatrix},
\end{equation}
where $J_{\alpha,\alpha}^{(n)} = \dfrac{\partial^2 l(\alpha, \bm{\beta} | \mathcal{Y}_n)}{\partial \alpha^2} $, $J_{\alpha,k}^{(n)} = J_{k, \alpha}^{(n)} = \dfrac{\partial^2 l(\alpha, \bm{\beta} | \mathcal{Y}_n)}{\partial \alpha \partial \beta_k }$, and $J_{k,k'}^{(n)} = \dfrac{\partial^2 l(\alpha, \bm{\beta} | \mathcal{Y}_n)}{\partial \beta_k \partial \beta_{k'}}; k,k' = 1, 2, \dots, K$, and $l(\alpha, \bm{\beta} | \mathcal{Y}_n) = \log[L(\alpha, \bm{\beta} | \mathcal{Y}_n)]$ is the log-likelihood under our model setup. 
\end{theorem}

The proof of Theorem \ref{th:Bayes_asymp} follows from the Bernstein-von Mises theorem and its generalized version under a linear model framework \cite{ghosal1999asymptotic}. The form of the log-likelihood, the derivatives, and the entries of $\mathcal{I}_n(\alpha, \bm{\beta})$ are derived in Appendix \ref{apndx:2}.

\section{Simulation Study}
\label{sec:Simulation}

We conduct an extensive simulation study to demonstrate (i) the effectiveness of the PC prior described in Section \ref{subsec:PC prior} over a conventional choice of a gamma prior and (ii) the effectiveness of the proposed semiparametric model, where the GE rate parameter (in log scale) is modeled as a linear or nonlinear function of the covariate(s), over a parametric model. As described in the following table, we consider eight simulation settings to address the simulation goals mentioned in (i) and (ii).

\begin{table}[ht]
\centering
\normalsize
\label{tab:sim-settings}
\begin{tabular}{l c c c}
\hline
\textbf{Setting} & \textbf{Generating data from} & \textbf{Fitted Model} & \textbf{Prior for $\alpha$} \\
\hline
Setting 1 & Linear set up     & Parametric     & Gamma \\
Setting 2 & Non-linear set up & Parametric     & Gamma \\
Setting 3 & Linear set up     & Parametric     & PC \\
Setting 4 & Non-linear set up & Parametric     & PC \\
Setting 5 & Linear set up     & Semiparametric & Gamma \\
Setting 6 & Non-linear set up & Semiparametric & Gamma \\
Setting 7 & Linear set up     & Semiparametric & PC \\
Setting 8 & Non-linear set up & Semiparametric & PC \\
\hline
\end{tabular}
\end{table}

The structure of each simulation setting is as follows. We first choose the sample size $n$ and generate the data $Y$ from a GE distribution, specifying the true values of $\alpha$ and $\lambda$. Next, we set the corresponding hyperparameters of the prior for $\alpha$ and fit either a parametric or semiparametric model using MCMC. Finally, we estimate several quantities, which are then used to assess the simulation goals. The detailed specifications are provided below.

For $n$, we consider two cases, $n$ = 24 and $n$ = 99, to gain insights into scenarios with small and large sample sizes, respectively. For $\alpha$, we choose it among one of three values, namely 0.5, 1, and 2, which allows us to investigate different scenarios; $\alpha=1$ represents the exponential distribution scenario, and in contrast, the values of $\alpha$ being 0.5 and 2 indicate deviations from the exponential-like behavior. When $\alpha=1$, the proposed PC prior would try to shrink the fitted model towards the true data-generating mechanism. However, when $\alpha=0.5$ or $\alpha=2$, the true-generating mechanism is not an exponential regression setting, and hence, the PC prior would shrink the fitted model towards a wrong model. Thus, our goal remains to judge the advantages of the PC prior when $\alpha=1$ and whether the prior provides a robust estimate when $\alpha \neq 1$. 

To compute $\lambda$, we need to specify both the covariates $(x_i)$ and the true $\beta$ values. The covariate matrix includes an intercept and two additional columns. The first column consists of an equally spaced sequence of values between 0 and 1. Specifically, for a sample size of $n=24$, the realized covariate values are $\mathcal{X} = (0.04, 0.08, \dots, 0.96)'$, while for a larger sample size of $n=99$, the values are $\mathcal{X} = (0.01, 0.02, \dots, 0.99)'$. The second column is generated as a random covariate drawn from a $\text{uniform}(0, 0.5)$ distribution. For the linear data-generating process, we use the covariate matrix as defined, with the $i$-th covariate given as $x_i^{L} = [1, x_{i1}, x_{i2}]$. In the nonlinear case, we transform the second and third columns to get $x_i^{NL} = [1, x_{i1}^2, \sin(2\pi x_{i2})]$. Further, we specify $\beta_{\text{true}} = (\beta_0, \beta_1, \beta_2)$ as (-5, 5, 3) for the linear case and (-5, 2, 3) for the nonlinear case. 

In summary, when the true data generating scheme is linear, we simulate data following $Y_i | X_i = x_i \overset{\textrm{Indep}}{\sim} \textrm{GE}\big(\alpha, \lambda_i = \exp[-5 + 5 x_{i1} + 3 x_{i2}]\big),~i=1, \ldots, n$. Alternatively, when the true data generating scheme is nonlinear, data are generated from $Y_i | X_i = x_i \overset{\textrm{Indep}}{\sim} \textrm{GE}\big(\alpha, \lambda_i = \exp[-5 + 2x_{i1}^2 + 3 \sin(2\pi x_{i2})]\big),~i=1, \ldots, n$. 
Furthermore, to ensure that the random covariates remain constant across all scenarios, we set the random seed in \texttt{R} to \texttt{set.seed(100)} prior to sampling from the uniform distribution. Moreover, the choice of $\beta_{\text{true}}$ does not substantially affect the final results.

We consider two sets of hyperparameters for each prior specification of $\alpha$. The hyperparameter is set to either 2.5 or 5 for the PC prior, while for the gamma prior, it is set to either $(0.01, 0.01)$ or $(1, 1)$. To clarify, each of the settings described above is evaluated under 12 different design specifications, arising from the combination of two sample sizes ($n$), three values of $\alpha$, and two sets of hyperparameters for the prior on $\alpha$, totaling 96 different design specifications. For $\beta_i$'s, we use a non-informative Gaussian prior with zero mean and variance 100.

Two types of models are fitted to the generated datasets: parametric and semiparametric. In the parametric case, we use the same covariate matrix as in the data-generating process and estimate the corresponding $\beta$ coefficients. In the semiparametric case, $x_{i1}$ and $x_{i2}$ are expanded into two sets of B-splines with four basis functions each, and the spline coefficients are estimated.

We employ the MCMC algorithm to draw samples from the posterior distributions. We draw 25,000 MCMC samples for $\alpha$ and $\beta$ and discard the first 12,500 samples as burn-in. Further, we thin the chains, keeping one of five consecutive samples, and draw inferences based on the remaining 2,500 MCMC samples. We monitor the convergence and mixing of the MCMC chains using trace plots and summaries like effective sample sizes, and we obtain the desired results; we do not show the MCMC chain-related diagnostics corresponding to the simulation studies for brevity. 
Furthermore, to stabilize the variance of the results, we replicate each of the 96 design specifications 1,000 times.

We now describe the analysis and plots used to address our two simulation goals. For the first goal, which evaluates the efficacy of the PC prior, we compare the use of the gamma prior against the use of the PC prior across the four combinations of data generation and model fitting scenarios. Specifically, we compare Setting 1 with 3, Setting 2 with 4, Setting 5 with 7, and Setting 6 with 8. For each setting, we compute two performance measures: coverage probability and absolute fitting bias of $\alpha$. Coverage probability is calculated by checking whether the 95\% credible interval for $\alpha$ contains the true value and then averaging over the 1,000 replications. Absolute fitting bias is defined as the absolute difference between the true $\alpha$ and its posterior mean. We also average it over the 1,000 replications.

For the second goal, where we compare the efficacy of the semiparametric fit against the parametric fit, we contrast Settings 1–4 with Settings 5–8, respectively. We evaluate performance using WAIC and absolute fitting error. The WAIC is computed based on the deviance calculated from the fitted likelihood. The absolute fitting error is defined as the absolute difference between the true and estimated values of $\log(\lambda) = \log(x_i^\top \beta)$, summed over all data points. Importantly, we do not directly compare the estimated $\beta$ values with the true values, since in the semiparametric case the use of basis splines prevents recovery of the original $\beta$’s. Finally, both WAIC and absolute fitting error are averaged over 1,000 replications.

Figure \ref{fig:Sim_PC_cov prob_PC_AB} corresponds to the first simulation goal. 
The top illustration presents the coverage probabilities of different simulation setups. Four columns represent the four combinations of data generation and model fitting scenarios: Linear-Parametric, Nonlinear-Parametric, Linear-Semiparametric, Nonlinear-Semiparametric. The two rows represent sample sizes of 24 and 99, respectively. Each panel showcases the four prior (two gamma and two PC) specifications represented by different lines to facilitate the comparison. The X-axis represents the different true values of $\alpha$.
The illustration shows that when the true value of $\alpha$ approximates one, the PC prior exhibits superior coverage probability compared to the conventional gamma prior. This pattern holds across all four configurations, except the second one, where both the gamma and PC priors yield undesirable outcomes due to attempts to fit a parametric linear model to highly nonlinear data. 

In the bottom illustration of Figure \ref{fig:Sim_PC_cov prob_PC_AB}, the identical simulation setups are depicted, but for the absolute bias in the estimation of $\alpha$. This illustration highlights a reduction in estimation bias with the PC prior when the true $\alpha$ value is one, aligning with the inherent characteristic of the PC prior, to shrink the estimate towards the base model. The trend is particularly evident under the PC(5) prior. Additionally, as the sample size increases, the influence of the prior gradually diminishes; this pattern is evident as the lines representing absolute bias nearly overlap, regardless of the chosen prior.

\begin{figure}[t]
    \centering
    \includegraphics[width=\textwidth]{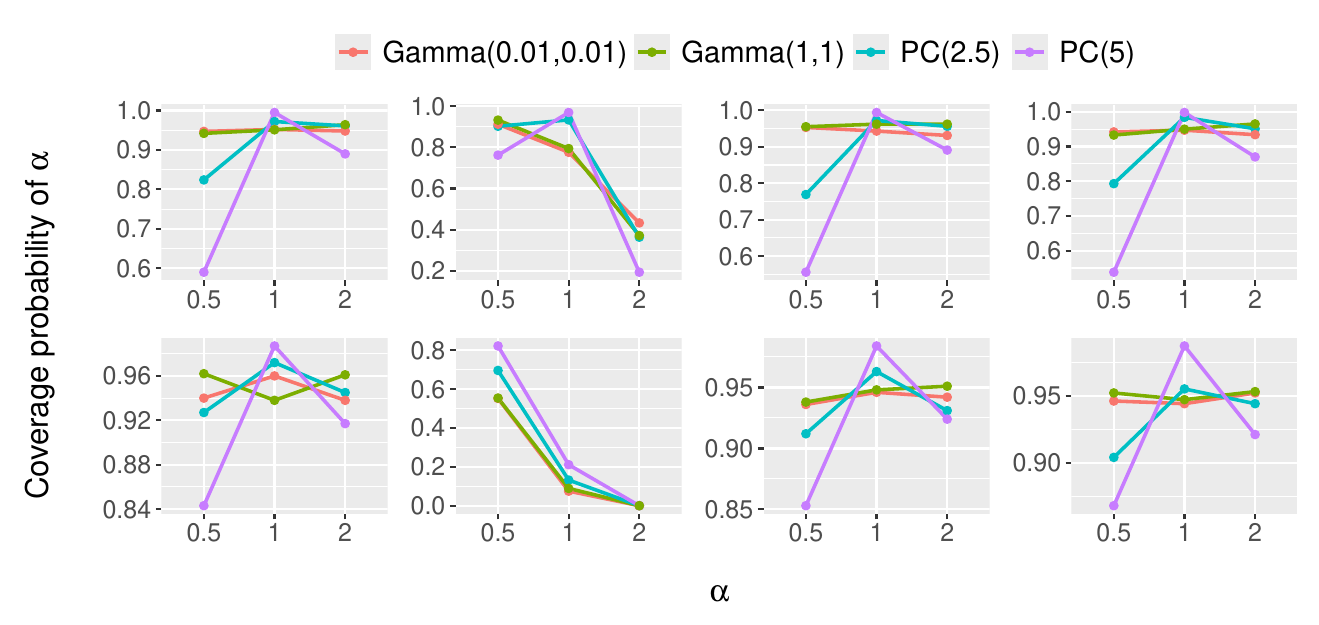}
    \includegraphics[width=\textwidth]{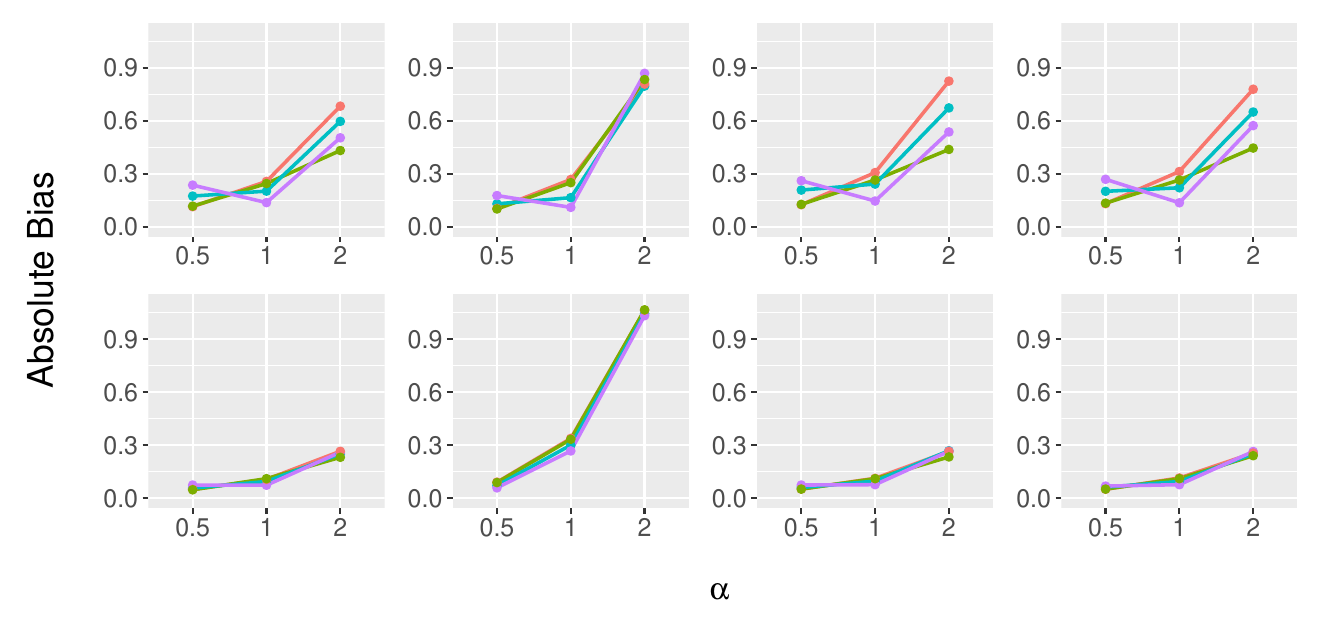}
    \vspace{-0.7cm}
    \caption{Coverage probabilities (top illustration) and absolute bias values (bottom illustration) computed based on imposing a PC prior. For each illustration, two rows depict $n=24$ and $n=99$, respectively, and four columns represent the four combinations of data generation and model fitting scenarios.}
    \vspace{-0.1cm}
    \label{fig:Sim_PC_cov prob_PC_AB}
    \end{figure}

Figure \ref{fig:Sim_semipar_WAIC_Abs error} focuses on the second simulation goal. In the top illustration, we compare the goodness of fit between the parametric and semiparametric models using the absolute fitting error. The four columns represent the four combinations of data generation and prior specifications scenarios: Linear-Gamma, Linear-PC, Nonlinear-Gamma, and Nonlinear-PC. Same as before, the two rows represent sample sizes 24 and 99, respectively. Each panel displays specifications of different values of $\alpha$ on the X-axis, and lines correspond to either the parametric setup or the semiparametric setup, with varying hyperparameter specifications. Although the plots generated under linear data do not provide conclusive evidence for the superiority of the semiparametric model, a distinct pattern emerges where the data is nonlinear. In these cases, a considerable gap appears between the parametric and semiparametric fits, with the semiparametric model consistently performing better, as reflected by its lower position in the plots.

The bottom panel of Figure \ref{fig:Sim_semipar_WAIC_Abs error} presents a similar comparison, this time based on WAIC, where smaller values indicate better out-of-sample predictive performance. The overall trend mirrors the previous findings: when the data are generated from a linear setup, the parametric model outperforms the semiparametric model. However, in the nonlinear case, the ordering of WAIC scores is reversed, with the semiparametric model consistently achieving lower values than the parametric model. This reversal highlights the clear advantage of using semiparametric modeling, particularly in highly nonlinear scenarios. Notably, the lines with the same prior but different values of hyperparameters are almost overlapping for this panel.

\begin{figure}[t]
    \centering
    \includegraphics[width=\textwidth]{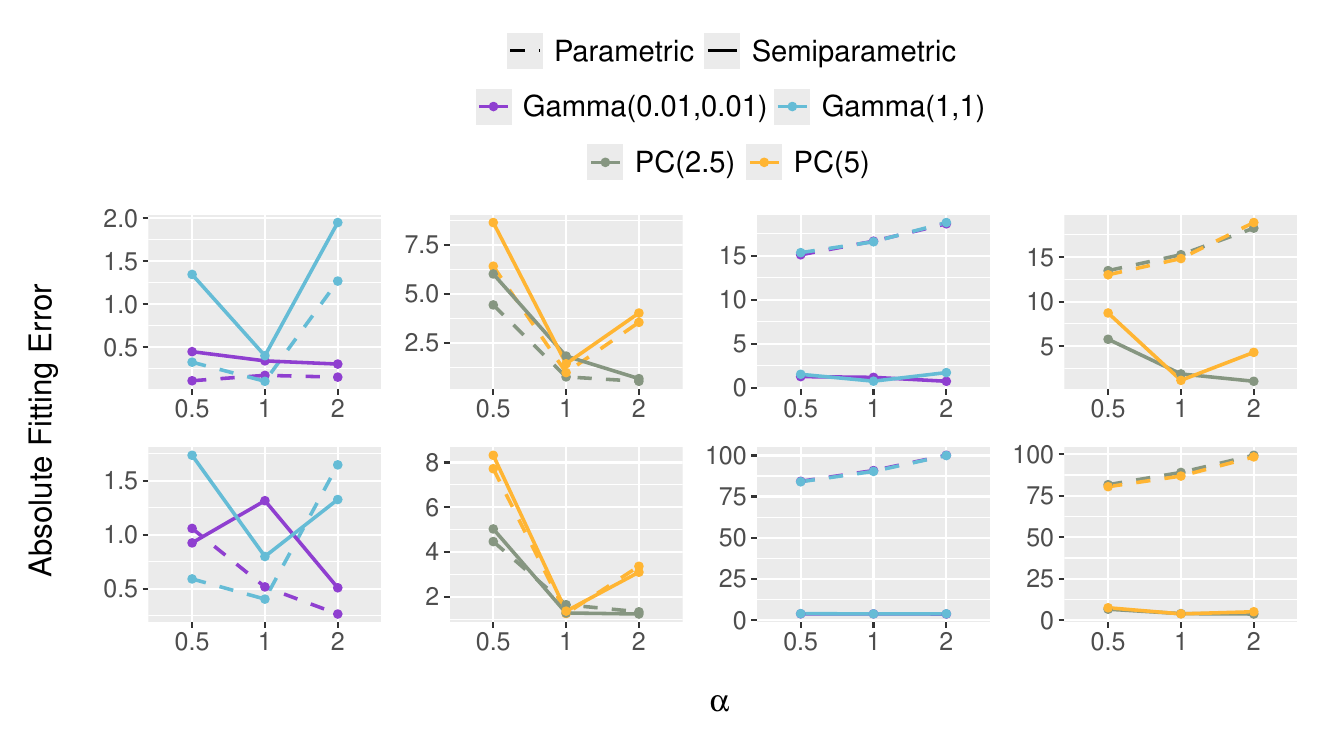}
    \includegraphics[width=\textwidth]{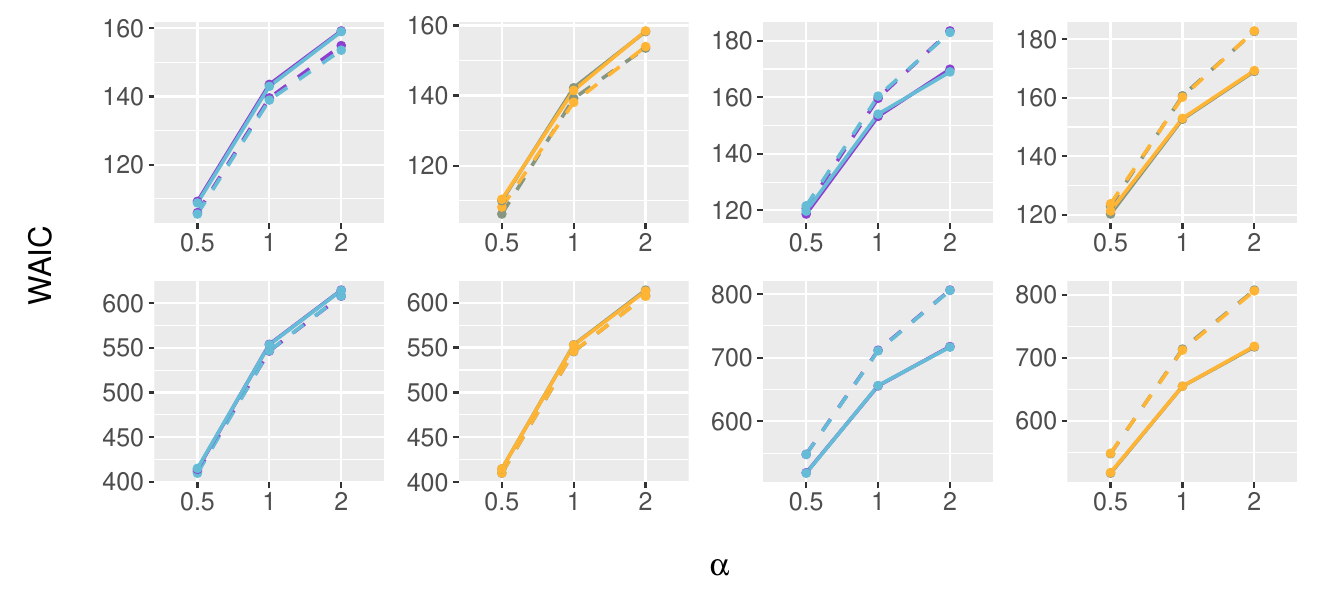}
    \vspace{-0.7cm}
    \caption{Absolute fitting error (top illustration) and WAIC values (bottom illustration) based on fitting a semiparametric GE regression model. For each illustration, two rows depict $n=24$ and $n=99$, respectively, and four columns represent the four combinations of data generation and prior specifications scenarios.}
    \vspace{-0.1cm}
    \label{fig:Sim_semipar_WAIC_Abs error}
\end{figure}

Regarding computational resources, the simulation study comprising 96 configurations with 1,000 replications and 25,000 MCMC samples per run required approximately 3.5 hours in total. The computations were carried out on a cloud system equipped with a 60-core processor and 72 GB of RAM, with parallelization applied across replications. On average, the computation time per setting (a single design specification) was about 1 minute for the parametric model and roughly 3 minutes for the semiparametric model, including all replications.


\section{Data Application}
\label{sec:Data_Application}

We obtain daily gridded rainfall (in mm) data over the Western Ghats of India region with a spatial resolution of $1.0\degree \times 1.0\degree$, covering the period from 1901--2022. The dataset is publicly available at the official website of the Indian Meteorological Department (IMD), Pune (\url{https://www.imdpune.gov.in/cmpg/Griddata/Rainfall_1_NetCDF.html}). The gridded data product was prepared by IMD through spatial interpolation of ground-station data following the procedure described in \cite{rajeevan2008analysis}. We extract the daily rainfall information for the monsoon months of June, July, August, and September (JJAS) throughout 1901--2022. Additionally, we exclude days within the JJAS months where recorded rainfall amounts were zero. Out of the pixels representing the Western Ghats area, we group them into three distinct significant regions: the Northern, Middle, and Southern regions (the regions are shown in the supplementary material). We compute the daily rainfall values for each region by calculating the average of the corresponding pixel values within that region. Afterward, we conduct a separate analysis for each of these regions.

\subsection{Exploratory Data Analysis}
\label{sec:Explore}

Given our dataset (after preprocessing) spans over a century, our initial focus involves performing necessary analyses to address any potential trends within the data. In the top panels of Figure \ref{fig:Exploratory_est_mean_hist dens}, we present a bar diagram depicting the average yearly rainfall for each year. No clear long-term linear trend is observable for any of the three regions. However, several short-term upward and downward patterns are noticeable. We use a basis spline regression approach to explore such short-term trends, which treats daily rainfall values as response variables and corresponding years as covariates. Considering the residuals from this regression, we can effectively eliminate any potential trends embedded in the data. We overlap the estimated means with the bar diagrams in the top panels of Figure \ref{fig:Exploratory_est_mean_hist dens}. Firstly, the estimated mean curve aligns well with the visualized bar diagram. Moreover, both components highlight the presence of a nonstationary rainfall pattern. This pattern, in turn, underscores the suitability of employing a semiparametric regression model, which can effectively accommodate and incorporate these nonstationary patterns within the data.

\begin{figure}[t]
    \centering
    \includegraphics[width=0.31\textwidth]{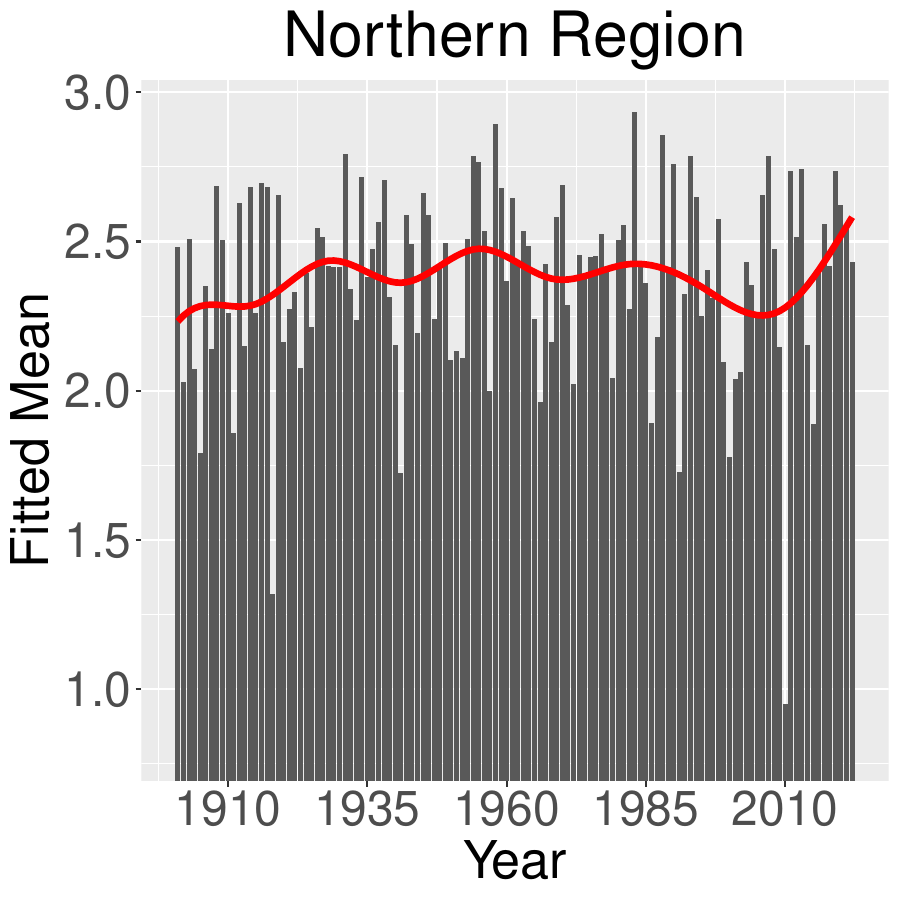}
    \includegraphics[width=0.31\textwidth]{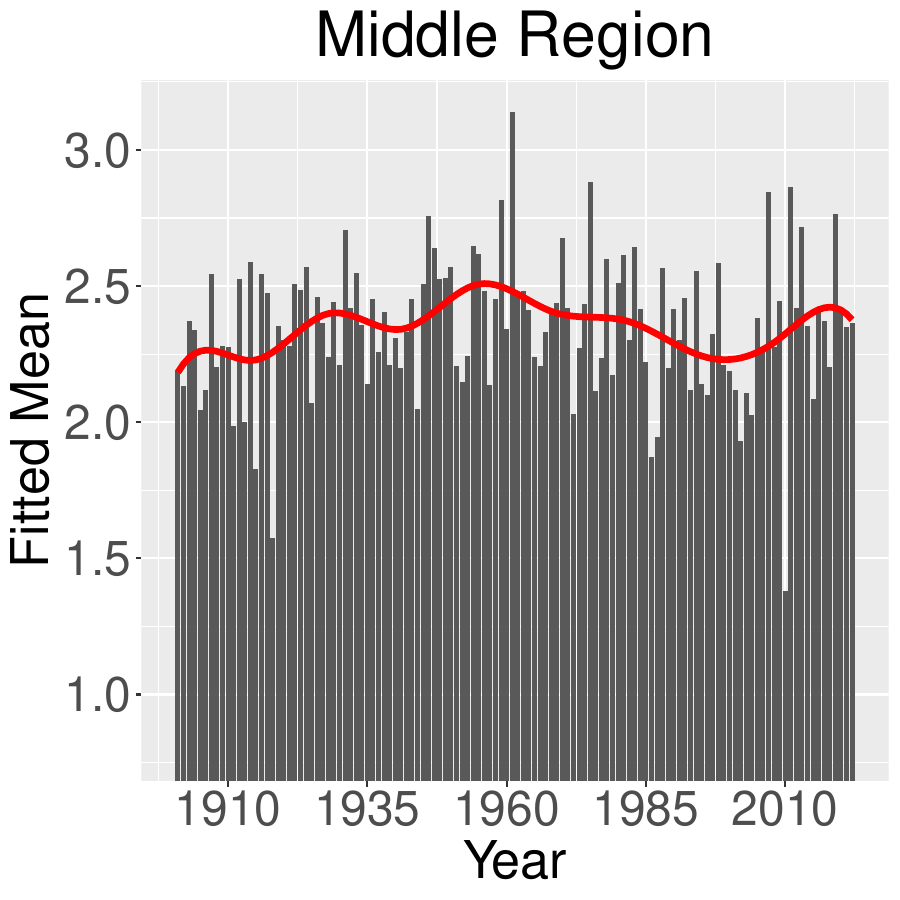}
    \includegraphics[width=0.31\textwidth]{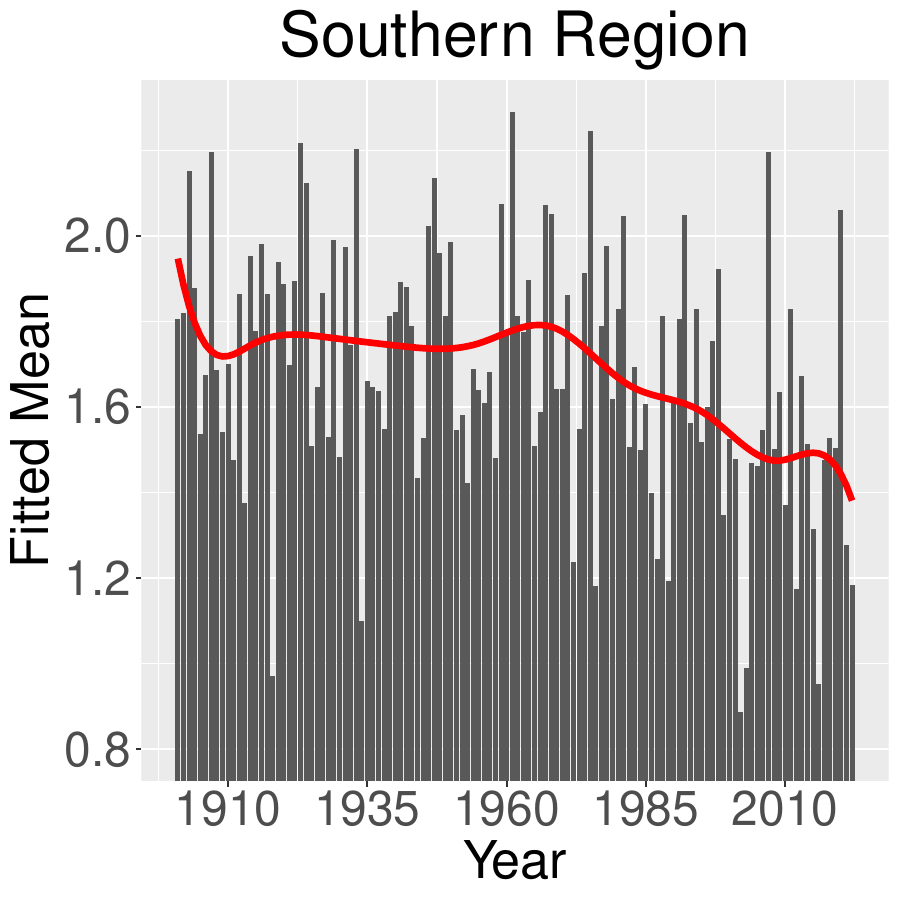} \\
    \includegraphics[width=0.31\textwidth]{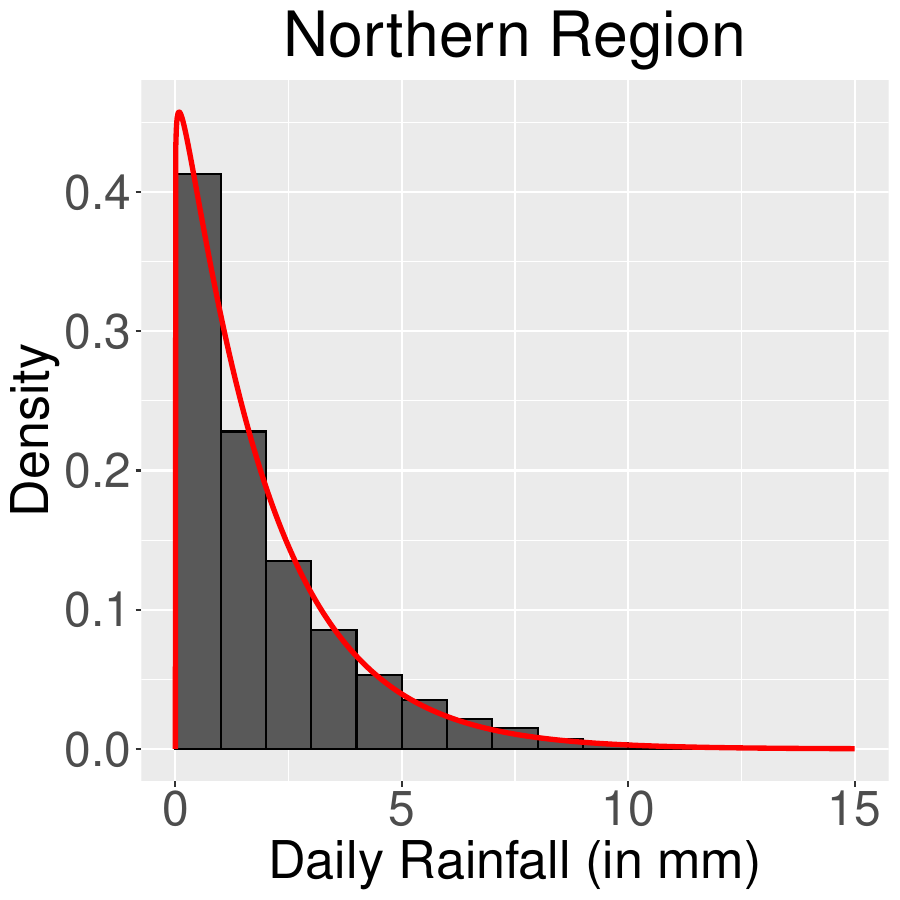}
    \includegraphics[width=0.31\textwidth]{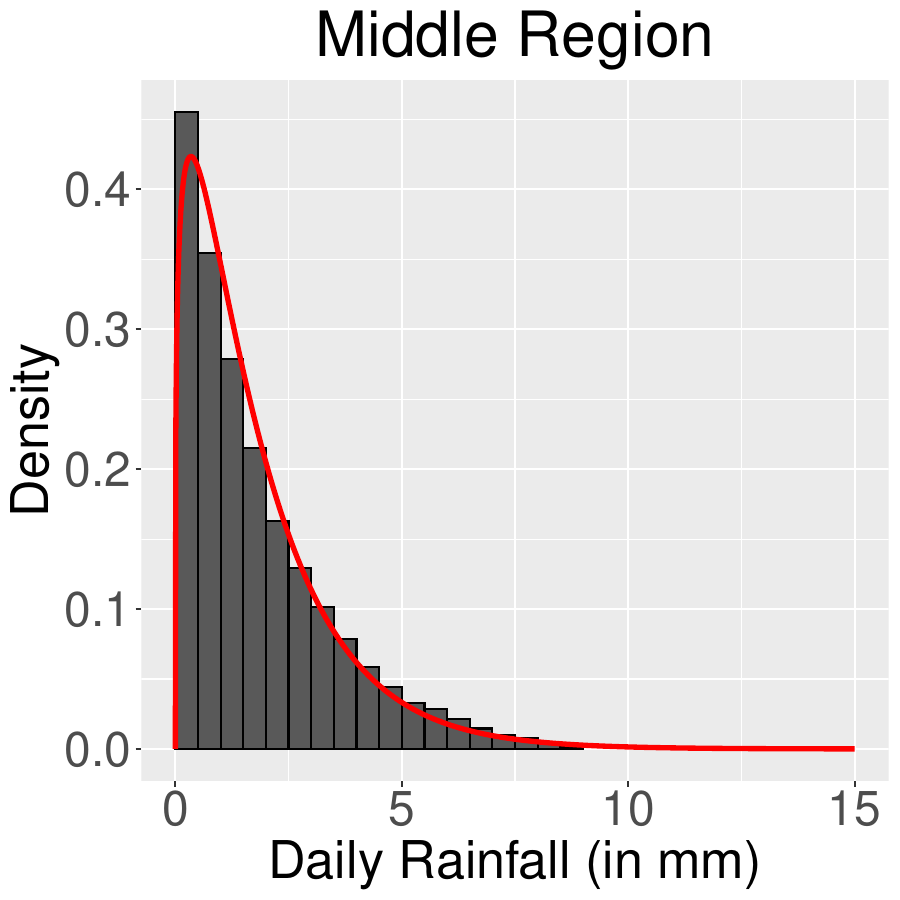}
    \includegraphics[width=0.31\textwidth]{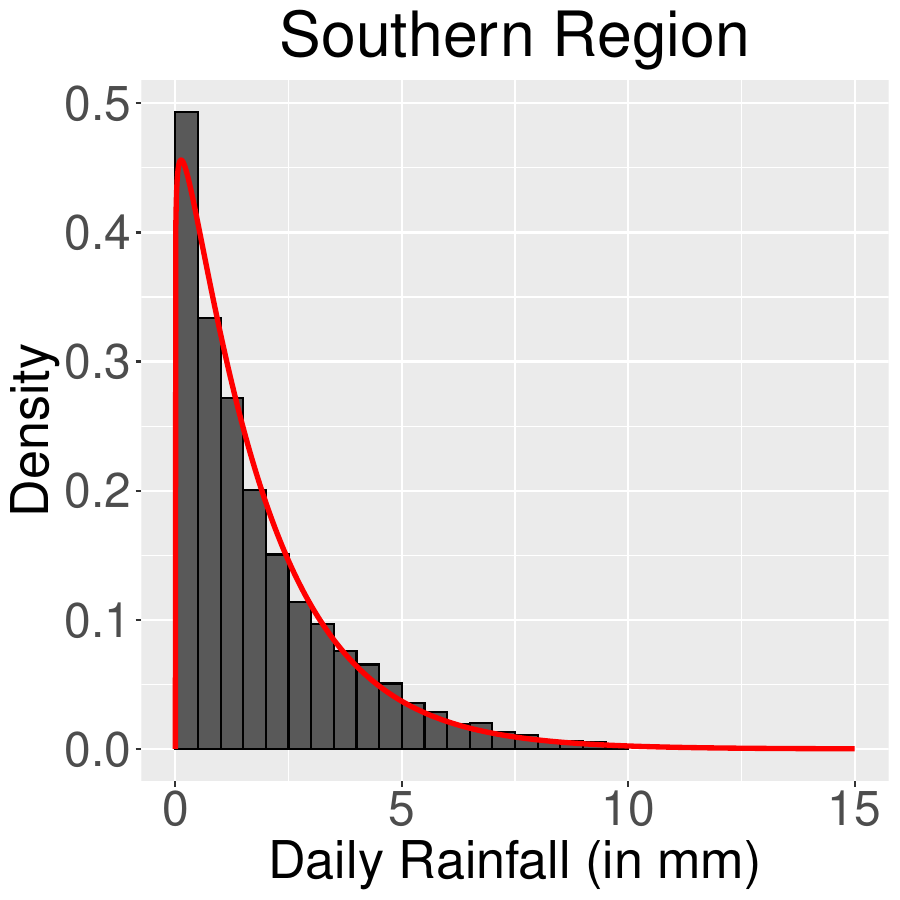} 
    \caption{Bar diagrams of the annual average wet-day rainfall during June through September along with fitted mean curves based on twelve cubic B-splines with equidistant knots, for Northern, Middle, and Southern Western Ghats regions (top panels). Histograms of the detrended residuals from the daily rainfall overlapped with the fitted GE densities (bottom panels).}
    \label{fig:Exploratory_est_mean_hist dens}
\end{figure}

Subsequently, after removing outliers via the popular adjusted-boxplot method developed by \cite{hubert2008adjusted}, we present two critical visualizations in the bottom panels of Figure \ref{fig:Exploratory_est_mean_hist dens}, where the panels correspond to three regions of interest. We first showcase histograms illustrating the distribution of the detrended residuals obtained by exponentiating the residuals obtained by fitting a semiparametric regression curve to the log-transformed wet-day rainfall observations, which aligns with the standard link function formulations for generalized additive models. Additionally, a red line denotes the fitted density of the GE distribution, with parameters estimated from the detrended residuals. We observe a strong alignment between the estimated density and the associated histograms, indicating a favorable fit. This visual representation significantly supports the rationale behind the semiparametric GE regression model proposed in this paper. Additionally, the plots highlight a marked resemblance of the GE distribution towards its foundational model, the exponential distribution, which has a density with mode at zero and then decays exponentially. This observation also reinforces our second consideration of using a novel distance-based prior for the shape parameter of the GE distribution.


\subsection{Model Specification}
\label{subsec:Model_descp}

In our case, except for the rainfall measures and the corresponding year information, we do not have access to additional covariates. Hence, we consider the year $T$ as a covariate in our analysis and assume it is defined over the continuous interval $[1901,2022]$, i.e. if $Y_t$ denotes a rainfall measure for year $t$, our model is given by $Y_t | T = t \overset{\textrm{Indep}}{\sim} \textrm{GE}(\alpha, \lambda(t)), ~~t=1901,\ldots,2022$. We conduct the analysis using two distinct models based on two different choices of $\lambda(t)$. The first model employs a parametric approach to model the rate parameter, while the second model is our proposed semiparametric formulation. We employ a simple linear regression model for the rate parameter in the parametric setting, given as $\lambda_{(\textrm{L})}(t)$ in \eqref{eq:data_app_two models}. On the other hand, for the semiparametric regression, we adopt the basis spline regression form presented in \eqref{eq:Basis splines} and choose $\lambda_{(\textrm{NL})}(t)$ in \eqref{eq:data_app_two models} as
\begin{equation}
\label{eq:data_app_two models}
\lambda_{\textrm{L}}(t)= \exp(\beta_0 + \beta_1  t), ~~~ \lambda_{\textrm{NL}}(t) = \exp\left[ \sum_{k = 1}^{K} \beta_k B_k(t) \right].
\end{equation}

We employ Bayesian methods to estimate the model parameters. As the exploratory analysis shows a strong resemblance of the considered GE distribution towards its base model, i.e., the exponential distribution, we use the proposed PC prior for the shape parameter $\alpha$ and choose independent weakly-informative Gaussian priors with mean zero and variance 100 for the regression parameters. As discussed in Section \ref{sec:Inference}, we use MCMC techniques to draw inferences about the model parameters. 

For our analysis, we use $K=12$ cubic B-spline basis functions with equidistant knots. With data available across 122 years, adopting 12 splines enables us to effectively capture decadal patterns using each spline \citep{hazra2020multivariate}. Here, the hyperparameter $\alpha_0$ for the PC prior being a tuning parameter, we compute WAIC values across a range of $\alpha_0$ values, spanning from 0.5 to 5 with a 0.5 increment, and achieve the most precise fit for our semiparametric regression. After individually examining the northern, middle, and southern regions, we identify the optimal values for $\alpha_0$ that yield the lowest WAIC values, and the optimal choices for these regions are $\alpha_0 = 4.5$, $\alpha_0 = 3.5$, and $\alpha_0 = 1.5$, respectively. As a result, these optimal $\alpha_0$ values are employed for their respective regions during the final model fitting stage.

For all three regions, we fit the two competing models in \eqref{eq:data_app_two models}. For each of the model fits, we generate 50,000 MCMC samples for each model parameter. The initial 15,000 samples are removed as burn-in and subsequently excluded from the analysis. Additionally, we employ a thinning interval of 5, and finally, we have 7,000 posterior samples for drawing inference. To evaluate the convergence and mixing of the chains derived from the MCMC process, we visualize the trace plots of the parameters associated with both model fits in the supplementary materials. Specifically, we present the trace plots of the shape parameter $\alpha$ for each region. The regression parameters also exhibit similar satisfactory mixing and convergence.


\subsection{Model Comparison}
\label{subsec:Model comp}

In this section, we compare the results based on the two models mentioned in \eqref{eq:data_app_two models}, and we present the estimated mean daily rainfall on wet days for each year between 1901 and 2022 for the Northern, Middle, and Southern Western Ghats regions and two competing models in Figure \ref{fig:model_fit}. Across all three regions, a noticeable trend emerges: the semiparametric models exhibit a notably superior fit. This distinction becomes evident as we observe multiple abrupt fluctuations in the bars representing the annual averages of wet-day precipitation. Remarkably, the semiparametric model effectively captures these fluctuations. Particularly noteworthy is the ability of the semiparametric model to accurately capture the nonstationarity present in the precipitation patterns. This heightened ability to encapsulate the dynamic variations in precipitation is a notable strength of the semiparametric model fitting. We also present the pointwise 95\% credible intervals for the trajectories estimated from the MCMC samples. The credible bands based on the semiparametric model are generally wider than the ones based on the parametric GE regression model; semiparametric models provide robust estimates in general but have higher uncertainty due to a bias-variance trade-off, and we observe the same here as well.

\begin{figure}[t]
    \centering
    \includegraphics[width=0.8\textwidth]{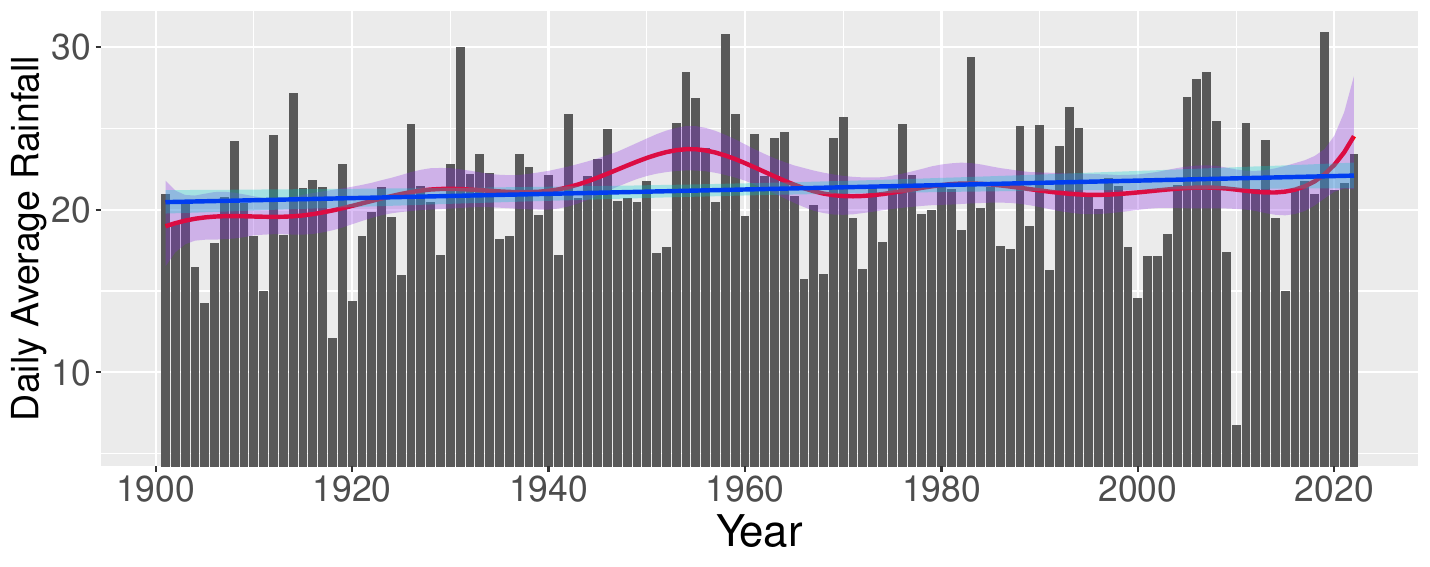}
    \includegraphics[width=0.8\textwidth]{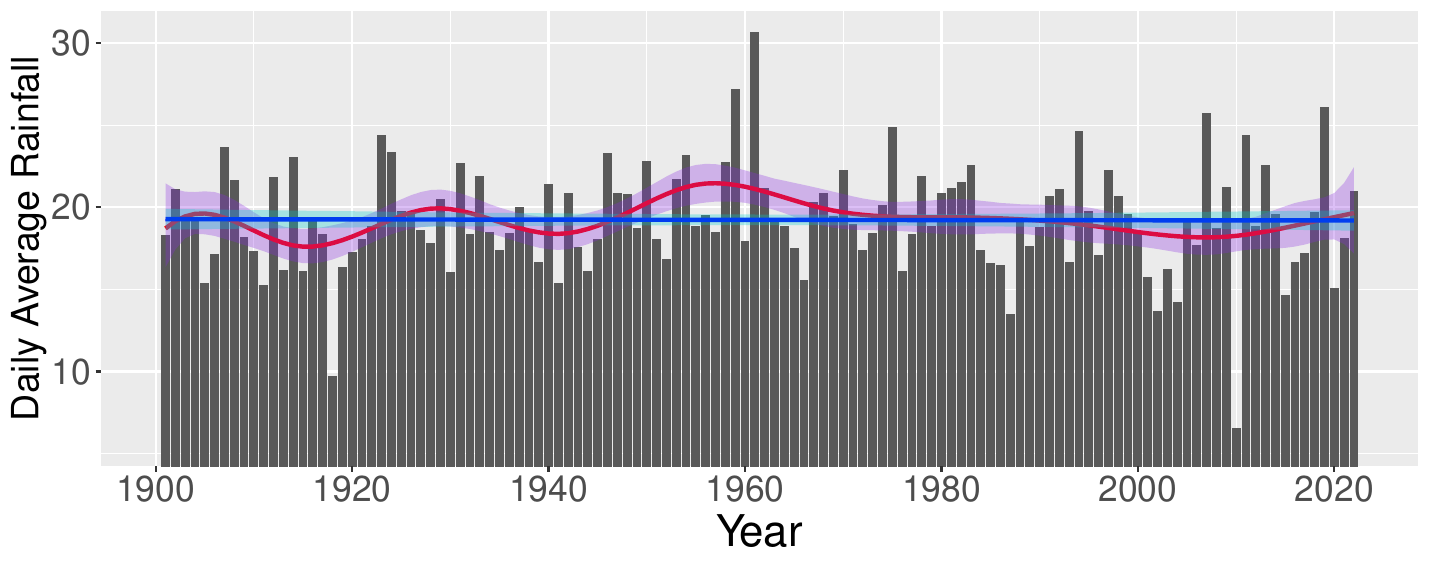}
    \includegraphics[width=0.8\textwidth]{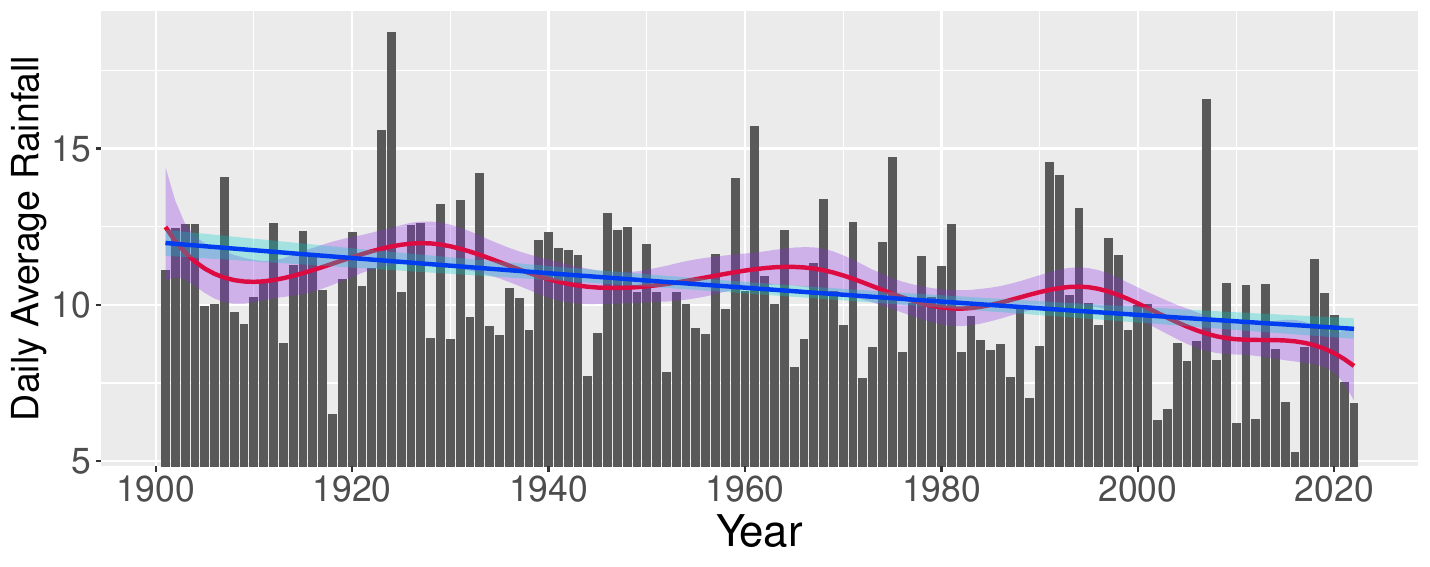}
    \caption{Estimated mean of daily wet-day rainfall (in mm) with semiparametric (red line) and parametric (blue line) models in \eqref{eq:data_app_two models}, along with corresponding pointwise 95\% credible intervals (ribbons). The top, middle, and bottom panels show the results for the Northern, Middle, and Southern Western Ghats regions.} 
    \vspace{-0.1cm}
    \label{fig:model_fit}
\end{figure}


\subsection{Inferences about Western Ghats Rainfall} 

In this subsection, we discuss results based on fitting the proposed GE regression model with a PC prior for the GE shape parameter. The estimated GE shape parameters (posterior means) for the Northern, Middle, and Southern Western Ghats regions are 0.859, 0.949, and 0.873, respectively, and the corresponding posterior standard deviations are 0.096, 0.100, and 0.097, respectively. These shape parameter values indicate a pronounced alignment with the exponential distribution of wet-day rainfall. Consistent with the fluctuating pattern in the annual average of daily wet-day rainfall, the fitted mean lines for each region also demonstrate short-term fluctuations. Besides, a consistent and stable mean rainfall trend is noticeable across the Northern and Middle Western Ghats regions. However, in the Southern Western Ghats region, the fitted parametric and semiparametric models distinctly reveal a decaying pattern in the annual averages of daily wet-day rainfall.

We present two significant insights into the rainfall patterns within these regions: the overarching decade-long shifts and individual region-specific probability rainfall plots. The calculation of the decadal change involves determining the overall rainfall shift and dividing it by the number of decades, resulting in $\{\mu(2022) - \mu(1901)\}/12.1$, where $\mu(t) = \lambda(t)^{-1} \left[ \psi(\alpha + 1) - \psi(1)\right]$, and $t$ representing the corresponding year; here, $\psi(\cdot)$ denotes the digamma function. Subsequently, the estimated decadal shifts in rainfall amount are 0.458 mm, 0.078 mm, and -0.367 mm for the northern, middle, and southern regions, respectively. In Figure \ref{fig:prob_rainfall}, we display the probability rainfall graphs for three distinct probabilities: 0.3 (red line), 0.5 (blue line), and 0.7 (green line); in agrometeorology, $100p$\% probability rainfall means the $(1-p)^{th}$ quantile of the probability distribution of rainfall. Figure \ref{fig:prob_rainfall} showcases the estimated probability-rainfall graphs, along with pointwise 95\% credible intervals for the estimated rainfall. We derive these intervals from the MCMC samples; they illustrate the uncertainty associated with the estimation process. For 70\% probability-rainfall, the pointwise credible bands are wider than the other two probability levels.

\begin{figure}[ht]
    \centering
    \includegraphics[width=0.32\textwidth]{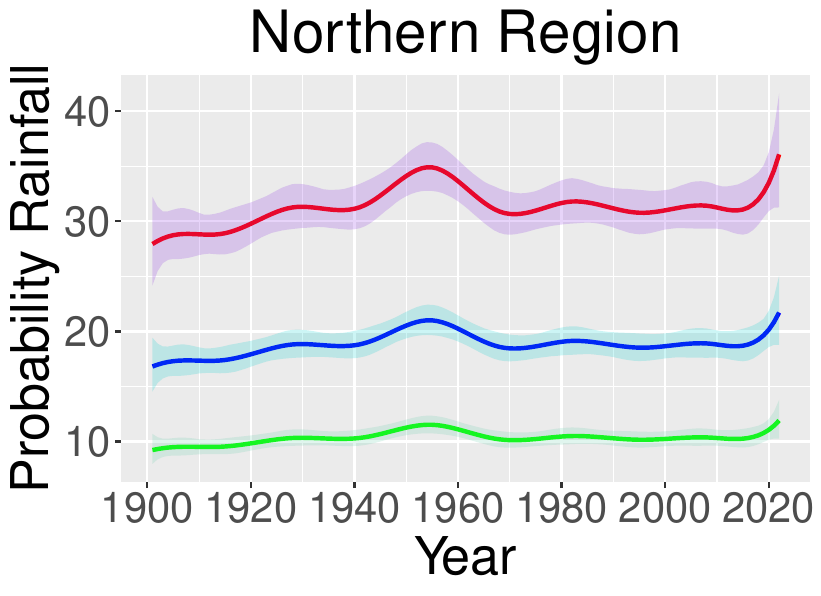}
    \includegraphics[width=0.32\textwidth]{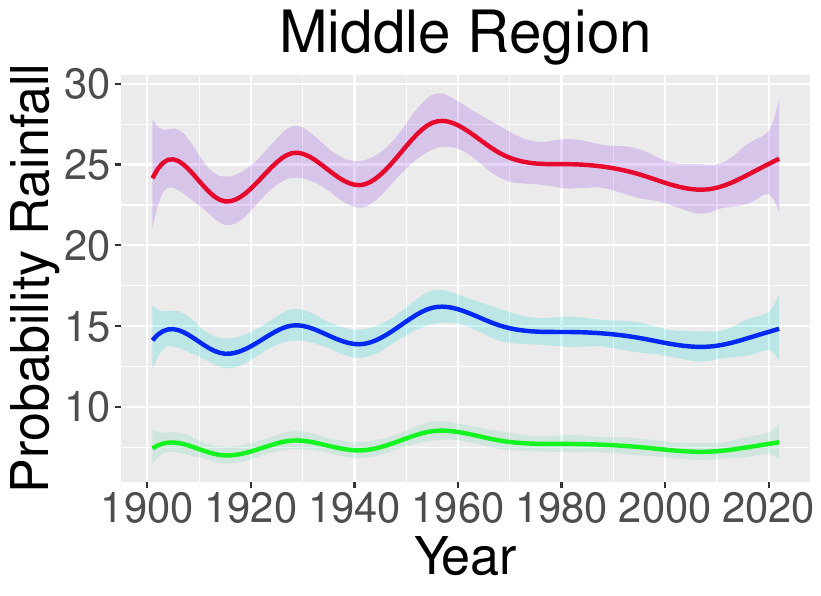}
    \includegraphics[width=0.32\textwidth]{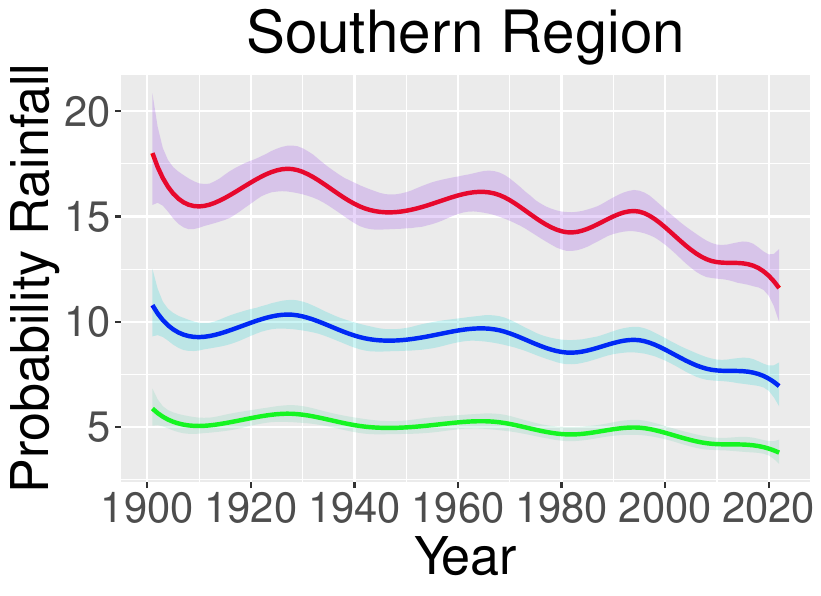}
    \caption{The 30\% (red line), 50\% (blue line), and 70\% (green line) probability rainfall (in mm) with corresponding pointwise 95\% credible intervals (ribbons).}
    \vspace{-0.1cm} 
    \label{fig:prob_rainfall}
\end{figure}

As a crucial component of our comprehensive analysis, we further discuss the dynamic nature of annual average rainfall for the Western Ghats range over the past century by exploring the plot of its rate of change. Interpreting this quantity unveils insights into trends, variations, and shifts in mean values over time, offering glimpses into rainfall behavior. A higher absolute magnitude implies fast changes, while a lower one indicates gradual shifts. A positive or negative rate means an increasing or decreasing mean rainfall over time, potentially signaling rising or decreasing annual average rainfall. We compute this quantity by taking the derivative of the fitted mean from our semiparametric model with respect to the time component from $\lambda_{\textrm{NL}}(t)$ in \eqref{eq:data_app_two models}, given by
\begin{equation}
\label{eq:change_of_rate}
    \frac{\partial \mu(t)}{\partial t} = -\frac{\psi(\alpha + 1) - \psi(1)}{[\lambda(t)]^2} \sum_{k = 1}^{K} \beta_k \frac{\partial B_k(t)}{\partial t},
\end{equation}
where we compute the derivatives of the cubic B-splines using \texttt{fda} package \citep{fdainR} in \texttt{R}.  

Figure \ref{fig:rate_of_change} illustrates varying trends in the rate of change in mean annual rainfall over the years across the three regions of the Western Ghats range. Initially, the Northern and Middle regions exhibit more pronounced fluctuations in the rate-of-change graphs than the Southern region. This pattern suggests more rapid variations in rainfall trends in the Northern and Middle regions, while a more stable rainfall pattern is visible for the Southern area. Moreover, the small-scale positive and negative rate-of-change instances are well-balanced for the Northern and Middle regions. This pattern implies that over the past century, changes in rainfall have been relatively symmetric in terms of increase and decrease, with no significant alterations in long-term patterns. In contrast, the Southern region displays a substantial portion of years with graphs below the zero line, signifying a prevalent decreasing trend in rainfall. The rate of change in mean for the last 30 years shows consistent negative values in the Southern sector, indicating the declining rainfall trend, while the graphs for the other two regions consistently exhibit positive values, indicating an increasing trend in rainfall over the past three decades in those areas. The pointwise 95\% credible intervals for the last 30 years include the zero line for the Northern and Middle regions; hence, the positive values for the last years are not significant. On the other hand, while the posterior mean rate-of-change remains negative for the Southern region in general, the credible intervals indicate that the negative values of rate-of-change are significant for several timestamps; however, the positive values are generally insignificant.

\begin{figure}[h]
    \centering
    \includegraphics[width=0.32\textwidth]{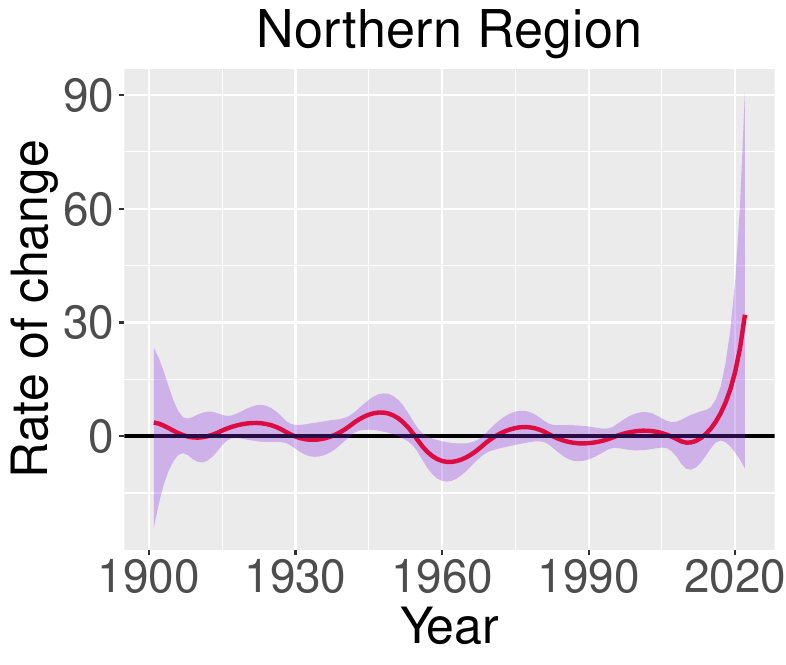}
    \includegraphics[width=0.32\textwidth]{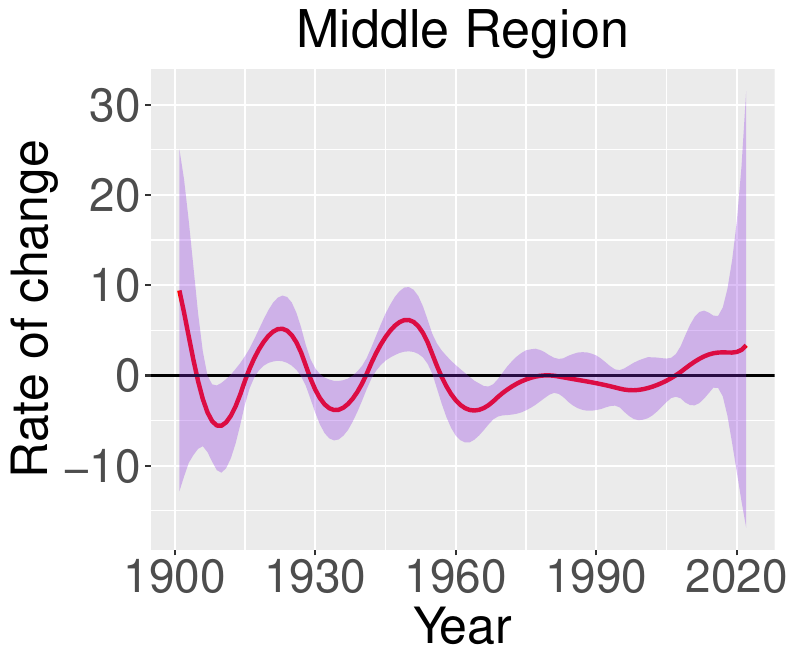}
    \includegraphics[width=0.32\textwidth]{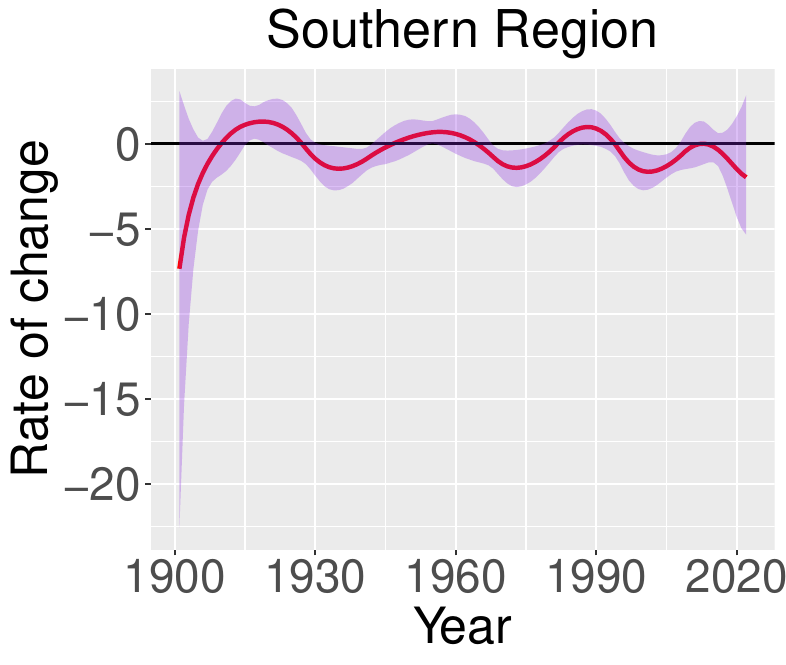}
    \caption{Rate of change in the annual average of daily wet-day rainfall in the monsoon months across the year 1901--2022, given by $\frac{\partial \mu(t)}{\partial t}$ in \eqref{eq:change_of_rate} (red line), and the corresponding pointwise 95\% credible intervals (ribbons). The black line represents the zero value.}
    \vspace{-0.1cm}
    \label{fig:rate_of_change}
\end{figure}


\subsection{Practical Significance of the Results} 

While the preceding subsections demonstrated the statistical superiority of the proposed semiparametric GE regression model and introduced some key quantities we explored, it is equally important to highlight the practical significance of the findings in the context of real-world rainfall applications. The ability of our model to flexibly capture both long-term trends and short-term fluctuations in rainfall has several direct implications.

First, detecting a significant declining trend in wet-day rainfall for the Southern Western Ghats has immediate consequences for regional agriculture and water resource management. Reduced rainfall in this region directly threatens crop yields, reservoir storage, and groundwater recharge, underscoring the urgency of developing adaptive strategies for irrigation planning and sustainable water use. In contrast, the relative stability of rainfall patterns in the Northern and Middle regions suggests that these areas may be more resilient to immediate climatic shifts, but continuous monitoring remains essential.

Second, by accurately modeling decadal fluctuations and nonstationary rainfall patterns, the semiparametric GE regression provides a more realistic representation of rainfall variability than traditional parametric alternatives. Such information is highly valuable for policymakers and planners who must allocate resources under uncertainty. For instance, reservoir release schedules, drought preparedness, and flood control strategies can all benefit from models that capture abrupt shifts in rainfall intensity over shorter horizons.

Lastly, the probability rainfall curves (e.g., 30\%, 50\%, and 70\% levels) derived from our Bayesian framework offer useful insights for agrometeorological planning. For example, in regions where the rainfall amount has been decreasing, such as the Southern Western Ghats, farmers may consider adjusting their cropping strategies — for instance, choosing varieties with longer harvesting periods so that the crop receives sufficient rainfall during its growth.


\section{Discussions and Conclusions}
\label{sec:Conclusion}

With its shape and rate parameters, the generalized exponential (GE) distribution facilitates more rigorous skewness attributes than several other distributions, specifically the exponential distribution. Thus, it is a better choice as a potential flexible model to incorporate high positive skewness in the data. Additionally, by varying the shape parameter, the hazard function of the GE distribution can adapt flexibly, making it more suitable for modeling complex data structures. In the regression arena, semiparametric regression is a powerful statistical method that combines the flexibility of nonparametric models with the interpretability and efficiency of parametric models. The superiority of our proposed model in capturing nonlinearity compared to the corresponding parametric model is depicted in Sections \ref{sec:Simulation} and \ref{sec:Data_Application}. On the other hand, penalized complexity (PC) prior is a principled distance-based prior that penalizes departure from a base model. It is used for specifying priors on parameters that are difficult to elicit directly from expert knowledge. This paper introduces a PC prior for the GE shape parameter, with the motivation of driving the GE distribution closer to the characteristics of the exponential distribution, a well-known probability distribution model for classical rainfall modeling.

The proposed semiparametric GE regression model reasonably fits the wet-day rainfall data for the Northern, Middle, and Southern regions of the Western Ghats mountain range of India. We observe a consistent overall trend with periodic fluctuations in the Northern and Middle Western Ghats regions. However, a declining trend is prominent in the Southern Western Ghats region. This observation is further supported by the decadal analysis of rainfall changes in these three regions, where only the Southern region exhibits a clear and significant negative value, indicating the effects of climate change. This research enhances our comprehension of the intricate climatic dynamics within the Western Ghats and emphasizes the critical role of precise predictive models in anticipating seasonal rainfall variations.

Alternative distributions, such as the gamma distribution, are also commonly used in rainfall modeling and remain effective in many applications. While gamma regression models are well established and often effective, our choice of the GE distribution offers some specific benefits. In particular, unlike the gamma distribution, the GE has a closed-form distribution function, which can simplify computations. Similar to the GE distribution, the natural base model for the gamma distribution is also the exponential distribution, once the shape parameter is set to one. However, the Kullback-Leibler divergence would involve gamma and digamma functions, and the final PC prior for the gamma shape parameter would also involve the trigamma function, making it difficult to study its theoretical properties explicitly. On the other hand, the PC prior for our GE case does not involve any special function, making it easier to note a Laplace distribution-type shape around the peak. Moreover, while gamma models perform well in standard settings, they may struggle with extreme rainfall events or complex hazard structures \citep{papalexiou2013extreme, hasan2020using}. The GE distribution, in contrast, provides additional flexibility in tail behavior and hazard shapes.

There are several directions for extending this research. In addition to modeling the rate parameter, we can consider treating the shape parameter as a time-dependent variable. Instead of utilizing splines for the rate parameter, an alternative approach could involve employing a Gaussian process prior. Moreover, to ensure the applicability of our comparisons to large datasets, we may explore various approximation techniques like Gaussian Markov random fields \citep{rue2005gaussian}. While this paper has primarily focused on the temporal analysis of rainfall data, further enhancements can be made by incorporating spatial components \citep{gotway1997generalized}. This extension involves investigating the variability in rainfall patterns across diverse geographical regions or watersheds \citep{yang2005spatial}. Additionally, there is potential for developing a real-time rainfall prediction system, offering timely information for tasks such as flood forecasting, reservoir management, and emergency response, based on the foundation provided by this model. For the high-dimensional spatial problems, our model can be implemented as a two-stage model where the GE parameters can be estimated at each spatial location, ignoring the spatial structure, and those estimates can be smoothed using a Gaussian process \citep{hazra2023bayesian}. 


\section*{Supplementary material}

Supplementary materials include the trace plots of the MCMC chains of the shape parameter of the generalized exponential regression model, under the proposed semiparametric model and the competing parametric model. We further provide a map identifying the Northern, Middle, and Southern Western Ghats regions on the map of India. Subsequently, we provide some details about selecting the hyperparameter of the proposed penalized complexity prior for the shape parameter of the generalized exponential regression model. The GitHub link for codes (written in \texttt{R}) for obtaining MCMC samples from the posterior distribution of the parameters of the proposed semiparametric GE regression model and the processed dataset analyzed in this paper is also provided.


\section*{Disclosure statement}

No potential conflict of interest was reported by the authors.


\begin{appendices}

\section{Kullback-Liebler Divergence between generalized exponential and exponential distributions}
\label{apndx:1}

We compute the Kullback-Liebler Divergence $\text{KLD}(\alpha) = \text{KLD}(f \parallel g)$, with $f$ being the generalized exponential density function in \eqref{eq:pdf_GE} and $g(y) = \lambda \exp[-\lambda y], y>0$. Thus, $\text{KLD}(\alpha)$ is obtained as
\begin{flalign*}
    \textrm{KLD}(\alpha) & = \int_{0}^{\infty} \log\left( \frac{f(y)}{g(y)} \right)  f(y) \, dy  \\
    & = \log(\alpha) \int_{0}^{\infty} f(y) \, dy + (\alpha - 1) \int_{0}^{\infty} \log\left[ 1 - \exp(-\lambda y) \right]  f(y) \, dy \\
    & = \log(\alpha) + (\alpha - 1) \int_{0}^{\infty} \log\left( 1 - \exp[-\lambda y] \right)  \alpha \lambda \left[ 1 - \exp(-\lambda y) \right]^{\alpha - 1}  \exp(-\lambda y) \, dy \\
    & = \log(\alpha) - (\alpha - 1) \int_{0}^{\infty} \alpha x \cdot \exp[-x(\alpha - 1)]  \exp(-x) \, dx \hspace{5pt} \left[\textrm{replacing}~ \exp(-x) =  1 - \exp(-\lambda y) \right]  \\
    & =  \log(\alpha) + \frac{(1 - \alpha)}{\alpha}.&&
\end{flalign*}


\section{Derivation of the necessary quantities for Bayesian asymptotics.}
\label{apndx:2}

\underline{Proof of \textbf{Theorem 5.1}}: \\

\noindent Doob's theorem \citep{doob1949application} explains the asymptotic behavior of the posterior density. According to this theorem, assuming the sampling model $P_{\bm{\theta}} $ with $ \bm{\theta} \in \bm{\Theta}$ to be identifiable in the sense that $\bm{\theta} \neq \bm{\theta}'$ implies $P_{\bm{\theta}} \neq P_{\bm{\theta}'}$, there exists $\bm{\Theta}_\ast \subseteq \bm{\Theta}$ with $\Pi(\bm{\Theta}_\ast) = 1$ such that for each $\bm{\theta}_\ast \in \bm{\Theta}_\ast$, if $\mathcal{Y}_n = \{Y_1, \dots, Y_n \}$ are iid $P_{\bm{\theta}_\ast}$, then for all $\epsilon > 0$, we have $$\lim_{n \to \infty} \mathbb{P} \left( \bm{\theta} \in \mathcal{N}_\epsilon \left(\bm{\theta}_\ast\right) | \mathcal{Y}_n\right) = 1,$$ where $\mathcal{N}_\epsilon(a) = \{ \bm{\theta} \in \bm{\Theta} : d(\bm{\theta},a) < \epsilon \}$, $d$ being a metric on $\bm{\Theta}$. The theorem implies that as long as the set of possible parameter values under consideration has a positive probability assigned to it by the prior distribution, the posterior distribution obtained after incorporating new information will tend to concentrate around the true parameter value.

The extension of Doob's theorem under a regression setting \cite{ghosal2017fundamentals}, where the covariates (splines in our case) are deterministic, the likelihood is identifiable, and the prior assigns positive weight throughout the parameter space, ensures the concentration of the posterior near the true parameter values. Given that our proposed GE regression model has finite deterministic covariates (splines do not scale to infinity), the likelihood function is clearly identifiable, the proposed PC prior for the shape parameter assigns positive density to all possible shape values, and the weakly-informative Gaussian priors for the regression coefficients assign positive density to all possible parameter values, all necessary conditions for the Doob's theorem under a regression setting holds and the posterior consistency is confirmed. \\

\noindent \underline{Derivation of $\widetilde{\mathcal{I}}(\alpha, \bm{\beta})$ in \textbf{Theorem 5.2}}: \\

\noindent Following the notations in Section \ref{subsec:theo_regression}, for $i=1,\ldots, n$, we have $Y_i | \bm{X}_i = \bm{x}_i \overset{\textrm{Indep}}{\sim} \textrm{GE}\big(\alpha, \lambda(\bm{x}_i)\big)$ with $\log[\lambda(\bm{x}_i)] = \bm{B}(\bm{x}_i)' \bm{\beta}$, where $\bm{B}(\bm{x}_i) = (b_{i1}, \ldots, b_{iK})'$ and $\bm{\beta} = (\beta_1, \ldots, \beta_K)'$. Hence, the likelihood function in \eqref{eq:GE_lklhd_reg} can be rewritten as $L(\alpha, \bm{\beta} | \mathcal{Y}_n) = \prod_{i = 1}^{n} f\big(Y_i; \alpha, \exp[\bm{B}(\bm{x}_i)' \bm{\beta}] \big)$. Hence, the explicit form of the log-likelihood is
$$l(\alpha, \bm{\beta} | \mathcal{Y}_n) = n \log(\alpha) + \sum_{i = 1}^{n} S_i + (\alpha - 1) \sum_{i = 1}^{n} \log(1 - E_i) - \sum_{i = 1}^{n} Y_i \exp(S_i),$$
with $S_i = \bm{B}(\bm{x}_i)' \bm{\beta}$ and $E_i = \exp[- Y_i \exp(S_i)]$. Notably, $\exp(S_i) = \lambda(x_i) = \lambda_i$ (say). Then, the first derivatives of the log-likelihood with respect to the parameters are computed as follows.
\begin{eqnarray}
\nonumber  \frac{\partial l(\alpha, \bm{\beta} | \mathcal{Y}_n)}{\partial \alpha} &=& \frac{n}{\alpha} + \sum_{i = 1}^{N} \log (1 - E_i), \\
\nonumber \frac{\partial l(\alpha, \bm{\beta} | \mathcal{Y}_n)}{\partial \beta_k} &=& \sum_{i=1}^{n} b_{ik} + (\alpha - 1)  \sum_{i = 1}^{n} \left[ \frac{ Y_i \exp(S_i) b_{ik} }{ E_i^{-1} - 1} \right] - \sum_{i=1}^{n} Y_i \exp(S_i) b_{ik}  ,~~k=1,\ldots,K.
\end{eqnarray}
Further, we compute the second derivatives as follows. For all $k,k'=1,\ldots,K$,
\begin{eqnarray}
\nonumber J_{\alpha, \alpha}^{(n)} &=& \frac{\partial^2 l(\alpha, \bm{\beta} | \mathcal{Y}_n)}{\partial \alpha^2} = -\frac{n}{\alpha^2}, \\
\nonumber J_{\alpha, k}^{(n)} &=& \frac{\partial^2 l(\alpha, \bm{\beta} | \mathcal{Y}_n)}{\partial \alpha \partial \beta_k } = \sum_{i = 1}^{n} \left[ \frac{ Y_i \exp(S_i) b_{ik} }{ E_i^{-1} - 1} \right],~~k=1,\ldots,K, \\
\nonumber J_{k, k'}^{(n)} &=& \frac{\partial^2 l(\alpha, \bm{\beta} | \mathcal{Y}_n)}{\partial \beta_{k'} \partial \beta_k} = (\alpha - 1) \sum_{i = 1}^{n} \Bigg[ Y_i b_{ik} \cdot \frac{\partial}{\partial \beta_{k'}} \bigg\{ \frac{\exp(S_i) }{E_i^{-1} - 1} \bigg\} \Bigg] - \sum_{i=1}^{n} \big[ Y_i b_{ik} b_{ik'} \exp(S_i) \big] , \\
\nonumber &=& (\alpha - 1) \sum_{i = 1}^{n} \Bigg[ Y_i b_{ik} \cdot \frac{\exp(S_i) b_{ik'}}{(E_i^{-1} - 1)^2} \cdot \left[ E_i^{-1}  - E_i^{-1} Y_i \exp(S_i) - 1 \right] \Bigg] - \sum_{i=1}^{n} \big[ Y_i b_{ik} b_{ik'} \exp(S_i) \big].
\end{eqnarray}
Hence, the elements of the information matrix are given by
\begin{eqnarray}
\nonumber I_{\alpha, \alpha}^{(n)} &=& -\mathbb{E}(J_{\alpha, \alpha}^{(n)}) = \frac{n}{\alpha^2}, \\
\nonumber  I_{\alpha, k}^{(n)} &=& -\mathbb{E}(J_{\alpha, k}^{(n)}) = -\sum_{i = 1}^{n} \exp(S_i) b_{ik} \underbrace{\mathbb{E}\left( \frac{Y_i}{1 - E_i^{-1}} \right)}_{\text{(I)}},~~k=1,\ldots,K, \\
\nonumber I_{k, k'}^{(n)} &=& -\mathbb{E}(J_{k, k'}^{(n)})  =  \sum_{i=1}^{n} b_{ik} b_{ik'} \exp(S_i) \mathbb{E}\big[ Y_i \big] \\
\nonumber && - (\alpha - 1) \sum_{i = 1}^{n} \exp(S_i) b_{ik'}b_{ik} \mathbb{E}\Bigg[ \frac{Y_i}{(E_i^{-1} - 1)^2} \cdot \left[ E_i^{-1}  - E_i^{-1} Y_i \exp(S_i) - 1 \right]\Bigg]  \\
\nonumber &=& \sum_{i=1}^{n} b_{ik} b_{ik'} \exp(S_i) \cdot \frac{1}{\lambda(x_i)} \big[\psi(\alpha+1) - \psi(1) \big] \\
\nonumber && - (\alpha - 1) \sum_{i = 1}^{n} \exp(S_i) b_{ik'}b_{ik} \Bigg[ \mathbb{E}\Bigg( \frac{Y_iE_i^{-1}}{(E_i^{-1} - 1)^2} \Bigg) - \mathbb{E}\Bigg( \frac{Y_i^2E_i^{-1} \exp(S_i)}{(E_i^{-1} - 1)^2} \Bigg) - \mathbb{E} \Bigg( \frac{Y_i}{(E_i^{-1} - 1)^2} \Bigg)\Bigg] \\
\nonumber &=& \big[\psi(\alpha+1) - \psi(1) \big] \sum_{i=1}^{n} b_{ik} b_{ik'} - (\alpha -1) \sum_{i=1}^{n} \exp(S_i) b_{ik'}b_{ik} \Big[ (\text{II}) + (\text{III}) + (\text{IV}) \Big] .
\end{eqnarray}

We calculate the four expectations (I), (II), (III), and (IV) separately. We use the transformation of variables and the standard formula for $\mathbb{E}[\log(X)]$ and $\mathbb{E}[\log^2(X)]$ where $X \sim $Beta($\alpha, \beta$), given as
$$\mathbb{E}[\log(X)] = \psi(\alpha) - \psi(\alpha+\beta), \text{ and } \mathbb{E}[\log^2(X)] = [\psi(\alpha) - \psi(\alpha+\beta)]^2+\psi^{(1)}(\alpha) - \psi^{(1)}(\alpha+\beta),$$
where $\psi(\cdot)$ and $\psi^{(1)}(\cdot)$ denote the digamma and trigamma functions, respectively. In the following, $B(\cdot, \cdot)$ denotes the beta function.
\vspace{2mm}

\noindent \underline{Expectations (I):}
\begin{align*}
   & \mathbb{E}\left( \frac{Y_i}{1 - E_i^{-1}} \right) = \int_{0}^{\infty} \frac{y}{1-e^{\lambda_iy}} \alpha \lambda_i\big(1-e^{-\lambda_iy}\big)^{(\alpha-1)} e^{-\lambda_iy}~dy
\end{align*}
Let $z = e^{-\lambda_i y}$, then 
$y = -\frac{\log z}{\lambda_i}$, and $\lambda_i e^{-\lambda_iy} dy = -dz.$
Further, when $y=0 \Rightarrow z=1$, and $y\to\infty \Rightarrow z=0$. Hence
\begin{align*}
   \mathbb{E}\left( \frac{Y_i}{1 - E_i^{-1}} \right) &= \alpha \int_{1}^{0} \frac{-\frac{\log z}{\lambda_i}}{1-1/z} (1-z)^{(\alpha-1)} ~ (-dz) = - \frac{\alpha}{\lambda_i} \int_{0}^{1}\log z \cdot z~(1-z)^{(\alpha-2)} ~ dz \\
   & = - \frac{\alpha}{\lambda_i} \cdot B(2,\alpha-1) \cdot [\psi(2)-\psi(\alpha+1)] = - \frac{\psi(2)-\psi(\alpha+1)}{\lambda_i (\alpha-1)}.  
\end{align*}

\noindent \underline{Expectations (II):}
\begin{align*}
    \mathbb{E}\Bigg( \frac{Y_iE_i^{-1}}{(E_i^{-1} - 1)^2} \Bigg)& = \int_{0}^{\infty} \frac{ye^{\lambda_iy}}{\big(e^{\lambda_iy}-1\big)^2} \alpha \lambda_i\big(1-e^{-\lambda_iy}\big)^{(\alpha-1)} e^{-\lambda_iy_i}~dy \\
   &  = \alpha \int_{1}^{0} \frac{-\frac{\log(z)}{\lambda_i} \cdot \frac{1}{z}}{(\frac{1}{z}-1)^2} (1-z)^{\alpha-1} (-dz) = -\frac{\alpha}{\lambda_i} \int_{0}^{1} \log(z) \cdot z~(1-z)^{\alpha-3} ~ dz \\
   & = - \frac{\alpha}{\lambda_i} \cdot B(2,\alpha-2) \cdot [\psi(2)-\psi(\alpha)] = - \frac{\alpha[\psi(2)-\psi(\alpha+1)]}{\lambda_i (\alpha-1)(\alpha-2)}. 
\end{align*}

\noindent \underline{Expectations (III):}
\begin{align*}
    \mathbb{E}\Bigg( \frac{Y_i^2E_i^{-1} \exp(S_i)}{(E_i^{-1} - 1)^2} \Bigg)& = \int_{0}^{\infty} \frac{y^2e^{\lambda_iy}\lambda_i}{\big(e^{\lambda_iy}-1\big)^2} \alpha \lambda_i\big(1-e^{-\lambda_iy}\big)^{(\alpha-1)} e^{-\lambda_iy_i}~dy \\
   &  = \alpha \int_{1}^{0} \frac{\frac{\log^2(z)}{\lambda^2_i} \frac{1}{z} \lambda_i}{(\frac{1}{z}-1)^2} (1-z)^{\alpha-1} (-dz) = -\frac{\alpha}{\lambda_i} \int_{0}^{1} \log^2(z) \cdot z~(1-z)^{\alpha-3} ~ dz \\
   & = - \frac{\alpha}{\lambda_i} \cdot B(2,\alpha-2) \cdot \big[\big(\psi(2)-\psi(\alpha)\big)^2 + \psi^{(1)}(2)-\psi^{(1)}(\alpha) \big] \\
   & = - \frac{\alpha\big[\big(\psi(2)-\psi(\alpha)\big)^2 + \psi^{(1)}(2)-\psi^{(1)}(\alpha) \big]}{\lambda_i (\alpha-1)(\alpha-2)}. 
\end{align*}

\noindent \underline{Expectations (IV):}
\begin{align*}
    \mathbb{E}\Bigg(\frac{Y_i}{(E_i^{-1} - 1)^2}\Bigg)& = \int_{0}^{\infty} \frac{y}{\big(e^{\lambda_iy}-1\big)^2} \alpha \lambda_i\big(1-e^{-\lambda_iy}\big)^{(\alpha-1)} e^{-\lambda_iy_i}~dy \\
   &  = \alpha \int_{1}^{0} \frac{-\frac{\log(z)}{\lambda_i}}{(\frac{1}{z}-1)^2} (1-z)^{\alpha-1} (-dz) = -\frac{\alpha}{\lambda_i} \int_{0}^{1} \log(z) \cdot z^2~(1-z)^{\alpha-3} ~ dz \\
   & = - \frac{\alpha}{\lambda_i} \cdot B(3,\alpha-2) \cdot [\psi(3)-\psi(\alpha+1)] = - \frac{2[\psi(2)-\psi(\alpha+1)]}{\lambda_i (\alpha-1)(\alpha-2)}. 
\end{align*}

Hence, finally, we get
\begin{align*}
    I_{\alpha, k}^{(n)} & = \frac{\psi(2)-\psi(\alpha+1)}{(\alpha-1)} \sum_{i = 1}^{n} b_{ik} ,~~k=1,\ldots,K, \text{ and} \\
    I_{k,k'}^{(n)}&  =\sum_{i=1}^{n} b_{ik}b_{ik'} \Big[\psi(\alpha+1) - \psi(1) + \frac{\alpha[\psi(2)-\psi(\alpha+1)]}{(\alpha-2)} + \\ 
    & \hspace{2cm} \frac{\alpha\big[\big(\psi(2)-\psi(\alpha)\big)^2 + \psi^{(1)}(2)-\psi^{(1)}(\alpha) \big]}{(\alpha-2)} + \frac{2[\psi(2)-\psi(\alpha+1)]}{(\alpha-2)} \Big].
\end{align*}

Further, denoting the matrix $\widetilde{\mathcal{I}}(\alpha, \bm{\beta})$ element-wise as follows \begin{equation}
\nonumber    \widetilde{\mathcal{I}}(\alpha, \bm{\beta}) = \begin{pmatrix}
\lim_{n \rightarrow \infty} n^{-1} I_{\alpha,\alpha}^{(n)} & \lim_{n \rightarrow \infty} n^{-1} I_{\alpha,1}^{(n)} & \cdots & \lim_{n \rightarrow \infty} n^{-1} I_{\alpha,K}^{(n)} \\
\lim_{n \rightarrow \infty} n^{-1} I_{1,\alpha}^{(n)} & \lim_{n \rightarrow \infty} n^{-1} I_{1,1}^{(n)} & \cdots & \lim_{n \rightarrow \infty} n^{-1} I_{1,K}^{(n)} \\
\vdots  & \vdots  & \ddots & \vdots  \\
\lim_{n \rightarrow \infty} n^{-1} I_{K,\alpha}^{(n)} & \lim_{n \rightarrow \infty} n^{-1} I_{K,1}^{(n)} & \cdots & \lim_{n \rightarrow \infty} n^{-1} I_{K,K}^{(n)} 
\end{pmatrix},
\end{equation}
we have 
\begin{eqnarray}
\nonumber \lim_{n \rightarrow \infty} n^{-1} I_{\alpha,\alpha}^{(n)} &=& \frac{1}{\alpha^2}, \\
\nonumber  \lim_{n \rightarrow \infty} n^{-1} I_{\alpha, k}^{(n)} &=& \widetilde{b}_k\frac{\psi(2)-\psi(\alpha+1)}{(\alpha-1)},~~k=1,\ldots,K, \\
\nonumber  \lim_{n \rightarrow \infty} n^{-1}  I_{k, k'}^{(n)} &=& \widetilde{b}_{k, k'} \Big[\psi(\alpha+1) - \psi(1) + \frac{\alpha[\psi(2)-\psi(\alpha+1)]}{(\alpha-2)} + \\ 
\nonumber    && \hspace{2cm} \frac{\alpha\big[\big(\psi(2)-\psi(\alpha)\big)^2 + \psi^{(1)}(2)-\psi^{(1)}(\alpha) \big]}{(\alpha-2)} + \frac{2[\psi(2)-\psi(\alpha+1)]}{(\alpha-2)} \Big].
\end{eqnarray}
where $\widetilde{b}_k = \lim_{n \rightarrow \infty} n^{-1} \sum_{i = 1}^{n} b_{ik}$ and $\widetilde{b}_{k, k'} = \lim_{n \rightarrow \infty} n^{-1} \sum_{i=1}^{n} b_{ik}b_{ik'}$ for $k,k'=1,\ldots, K$, and the existence of these limits holds from the regularity conditions of Theorem 5.2.
\end{appendices}


\bibliography{sn-bibliography}

\end{document}